\documentclass[preprint,12pt]{elsarticle}

\usepackage{amssymb}

\let\today\relax
\makeatletter
\def\ps@pprintTitle{%
    \let\@oddhead\@empty
    \let\@evenhead\@empty
    \def\@oddfoot{\footnotesize\itshape
       \hfill\today}%
    \let\@evenfoot\@oddfoot
    }
\makeatother

\usepackage{xcolor}
\usepackage{hyperref}
\usepackage{amsmath,amssymb}
\usepackage{stmaryrd}
\usepackage{xspace}

\usepackage{mathtools}
\usepackage[ruled,vlined]{algorithm2e}
\usepackage{tikz}
\usetikzlibrary{automata,positioning,decorations.pathreplacing}

\usepackage{amsmath}
\usepackage{amsthm}

\usetikzlibrary{arrows,automata}

\usepackage{booktabs}

\usepackage{tikz-qtree}
\usepackage{caption}
\usepackage{subcaption}
\usepackage{sidecap}

\usepackage{forest}
\newcommand{\OMIT}[1]{}

\newcommand{\vl}{\, \, \vline \, \, }
\newcommand{\bz}{\mathbf{z}}
\newcommand{\bw}{\mathbf{w}}
\newcommand{\by}{\mathbf{y}}
\newcommand{\bx}{\mathbf{x}}

\newcommand{\br}{\mathbf{r}}

\newcommand{\ta}{\mathtt{a}}
\newcommand{\tb}{\mathtt{b}}

\newcommand{\clean}{\mathsf{clr}} 
 
\newcommand{\wild}{^*}
\newcommand{\rlang}{\mathcal{R}}

\newcommand{\refl}{\mathsf{Ref}}

\newcommand{\cdoto}{\mathop{\cdot}}

\newcommand{\doc}{\mathbf{d}}
\newcommand{\alphabet}{\Sigma}

\newcommand{\mspan}[2]{\ensuremath{[#1,#2\rangle}}

\newcommand{\vars}{\mathsf{Vars}}

\newcommand{\df}{:=}
\newcommand{\join}{\bowtie}

\newcommand{\p}{^\prime}
\newcommand{\repspnr}[1]{\llbracket{#1}\rrbracket}

\newcommand{\nter}{V}
\newcommand{\ter}{{\Sigma}}
\newcommand{\prodrules}{{P}}
\newcommand{\initvar}{{S}}

\newcommand{\nterdec}{V^\text{\sc{dec}}}

\newcommand{\prodrel}{ \Rightarrow }
\newcommand{\moverel}{\vdash }
\newcommand{\grmr}[1]{\text{\sc{#1}}}

\newcommand{\decgrmr}[1]{\text{\sc{decorGrmr}}(#1)}

\newcommand{\jmp}{\text{\sc{jump}}}
\newcommand{\vop}[1]{ \vdash_{#1}}
\newcommand{\vcl}[1]{ \dashv_{#1}}

\newlength\boxwidth
\newlength\questionwidth

\newcommand{\PDA}{\mathsf{PDA}}

\newcommand{\sel}{\zeta}
\newcommand{\diff}{\setminus}

\newcommand{\map}{\mathsf{map}}
\newcommand{\enum}{\text{\sc{enumerate}} }
\newcommand{\applyProd}{\text{\sc{applyProd}} }
				
\newcommand{\tsc}[1]{\text{\sc{#1}}}

\newtheorem{example}{Example}[section]
\AtEndEnvironment{example}{\hfill \qed}%

\newtheorem{definition}{Definition}[section]
\newtheorem{proposition}{Proposition}[section]
\newtheorem{lemma}{Lemma}[section]
\newtheorem{theorem}{Theorem}[section]
\newtheorem{corollary}{Corollary}[section]

\journal{Discrete Applied Mathematics}

\begin{document}

\begin{frontmatter}



\title{Enumerating Grammar-Based Extractions}


\author{Liat Peterfreund\corref{l}}
\address{CNRS, LIGM, Gustave Eiffel University, Paris, France}
\cortext[l]{Part of this work was done while affiliated with DI ENS, ENS, CNRS, PSL University, and Inria.}

\begin{abstract}
\sloppy{
We propose a new grammar-based language for defining information extractors from documents (text) that is built upon the well-studied framework of document spanners for extracting structured data from text. While previously studied formalisms for document spanners are mainly based on regular expressions, we use an extension of context-free grammars, called {extraction grammars}, to define the new class of context-free spanners. 
Extraction grammars are simply context-free grammars extended with variables that capture interval positions of the document, namely spans.
While regular expressions are efficient for tokenizing and tagging, context-free grammars are also efficient for capturing structural properties. Indeed, we show that context-free spanners are strictly more expressive than their regular counterparts. We reason about the expressive power of our new class and present a pushdown-automata model that captures it.
We show that extraction grammars can be evaluated with polynomial data complexity. Nevertheless, as the degree of the polynomial  depends on the query, we present an enumeration algorithm for unambiguous extraction grammars that, after quintic preprocessing, outputs the results sequentially, without repetitions, with a constant delay between every two consecutive ones. }
\end{abstract}




\end{frontmatter}

\section{Introduction}\label{sec:intro}
The abundance and availability of valuable textual resources in the last decades position
text analytics as a standard component in data-driven workflows.
One way to facilitate the analysis and integration of textual content is to extract
 structured data from
it, an operation we refer to as \emph{Information Extraction (IE)}.  
IE arises in a large variety of domains, including social media
analysis~\cite{DBLP:conf/acl/BensonHB11}, health-care
analysis~\cite{DBLP:journals/jamia/XuSDJWD10}, customer relationship
management~\cite{DBLP:conf/www/AjmeraANVCDD13}, information
retrieval~\cite{DBLP:conf/www/ZhuRVL07}, and more. 


\emph{Rules} have always been a key component in various paradig\-m\-s for IE, and their roles have varied and evolved over the time. 
Systems such as Xlog~\cite{DBLP:conf/vldb/ShenDNR07} and IBM's
SystemT~\cite{DBLP:conf/acl/LiRC11, DBLP:conf/acl/ChiticariuKLRRV10} use rules to extract 
relations from text (e.g., tokenizer,
dictionary lookup, and part-of-speech tagger) that are further manipulated with relational query languages. Other
systems use rules to generate features for machine-learning classifiers~\cite{DBLP:conf/dsmml/LiBC04,DBLP:journals/sigmod/SaRR0WWZ16}. 

\paragraph*{Document Spanners}
The framework of document spanners, presented by Fagin et al., provides a theoretical basis for investigating the principles of relational rule systems for IE~\cite{DBLP:journals/jacm/FaginKRV15}. 
The research on document spanners has roughly focused on their expressive power~\cite{DBLP:journals/jacm/FaginKRV15, DBLP:conf/icdt/Freydenberger17, DBLP:conf/icdt/PeterfreundCFK19, DBLP:journals/mst/FreydenbergerH18, DBLP:conf/webdb/NahshonPV16, DBLP:journals/corr/abs-1912-06110, schmid2020purely, peterfreund2022weight} their computational complexity~\cite{DBLP:conf/icdt/AmarilliBMN19,DBLP:conf/pods/FlorenzanoRUVV18, DBLP:conf/pods/FreydenbergerKP18, DBLP:conf/pods/PeterfreundFKK19, freydenberger2021splitting}, extensions that extract incomplete information~\cite{DBLP:conf/pods/MaturanaRV18, DBLP:conf/pods/PeterfreundFKK19}, system aspects such as cleaning~\cite{DBLP:journals/tods/FaginKRV16},
dynamic complexity~\cite{DBLP:conf/icdt/FreydenbergerT20},
distributivity~\cite{DBLP:conf/pods/DoleschalKMNN19, doleschal2021optimization}, their evaluation over compressed documents~\cite{schmid2021spanner,schmid2022query,munoz2022constant}, and 
logic based approaches~\cite{thompson2022conjunctive, DBLP:conf/icalp/FreydenbergerP21}.
For a broad overview of recent advancements, we refer the reader to~\cite{doleschal2021database, schmid2022document}


In the documents spanners framework, a
\emph{document} $\doc$ is a string over a fixed finite alphabet, and a \emph{spanner} is a function that extracts from a document a relation over the spans of $\doc$. A \emph{span} $x$ is a half-open interval of positions of $\doc$ and it represents a substring $\doc_x$ of $\doc$ that is identified by these positions. A natural way to specify a spanner is by a \emph{regex formula}: a regular expression with embedded \emph{capture variables} that are viewed as relational attributes.
For instance, the spanner that is given by the regex formula 
$(\ta \vee \tb)^* \vop{x} \ta \ta^* \vcl{x} \vop{y}  \tb \tb^* \vcl{y} (\ta \vee \tb)^*$ extracts from documents spans $x$ and $y$ that correspond, respectively, to a non-empty substring of $\ta$'s followed by a non-empty substring of $\tb$'s. In particular, it extracts from the document $\ta \tb \ta \tb \tb$ the relation depicted in Figure~\ref{fig:introb}.

\begin{figure}
		\centering
		\begin{center}
			\begin{tabular}{cc}
				\toprule
				x & y \\
				\toprule
				$\mspan{1}{2}$ & $\mspan{2}{3}$\\
				$\mspan{3}{4}$ & $\mspan{4}{5}$\\
				$\mspan{3}{4}$ & $\mspan{4}{6}$\\
				\midrule
			\end{tabular}
			\vspace{-0.4cm}
		\end{center}
		\caption{ {The relation extracted from the document $\ta \tb \ta \tb \tb$ by the spanner defined by $(\ta \vee \tb)^* \vop{x} \ta \ta^* \vcl{x} \vop{y}  \tb \tb^* \vcl{y} (\ta \vee \tb)^*$. 
  }}
		\label{fig:introb}
\end{figure}

The class of \emph{regular spanners} is the class of spanners definable as the closure of regex formulas under positive relational algebra operations: projection, natural join and union. The class of regular spanners can  be represented alternatively by finite state machines, namely \emph{variable-set automata (vset-automata)}, which are nondeterministic finite-state automata that can open and close variables (and that, as in the case of regex formulas, play the role of the attributes of the extracted relation).
\emph{Core} spanners~\cite{DBLP:journals/jacm/FaginKRV15} 
are obtained by extending the class of regular spanners with string-equality selection on span variables, and thus 
are strictly more expressive than their regular counterpart.  

To date, most research on document spanners has been focused on the regular representation that is based on regular expressions and finite state automata.
Regular expressions are useful for parsing tasks that involve simple patterns of strings. These tasks include segmentation and tokenization of simple data formats such as CSV files or log files. Nevertheless, regular languages fall short in  intricate tasks such as syntax highlighting~\cite{moro2001syntax} and finding patterns in source code~\cite{smith2003spqr}. Generally, for parsing tasks that involve more complex structures such as those found in many programming languages, and for nested data formats such as XML and JSON, context-free languages are necessary. 

It is well known that context-free languages are strictly more expressive than regular languages. 
B\"{u}chi~\cite{buchi} has shown that regular languages are equivalent to monadic second order logic (over strings), and Lautemann et al.~\cite{DBLP:conf/csl/LautemannST94} have shown that adding an existential quantification over a binary relation interpreted as a matching is enough to express all context-free languages.
This quantification, intuitively, is what makes it possible to also express structural properties.

\paragraph*{Contribution}
In this work we propose a new grammar-based approach for defining the class of \emph{context-free spanners}. Context-free spanners are defined via \emph{extraction grammars} which, like regex formulas, incorporate \emph{capture variables} that are viewed as relational attributes. Extraction grammars produce \emph{ref-words} which are words over an extended alphabet that consists of standard terminal symbols along with \emph{variable operations} that denote opening and closing of variables. 
The result of evaluating an extraction grammar on a document $\doc$ is defined
via the ref-words that are produced by the grammar and equal to $\doc$ after erasing the
variable operations.

\begin{figure}[t]
\begin{center}
	\begin{minipage}[b]{0.3\textwidth}
		\centering
		\captionsetup{justification=centering}
		\begin{tabular}{ l } 
			$	{S}\rightarrow B \vop{x} \ta  A \tb \vcl{y} B $ \\ 
			$A\rightarrow\ta  A \tb  \vl \vcl{x}  \vop{y}$  \\
			$B\rightarrow \ta B \vl \tb B \vl \epsilon$
		\end{tabular}
	\end{minipage}
	\begin{minipage}[b]{0.3\textwidth}
 	\end{minipage}
	\begin{minipage}[b]{0.3\textwidth}
		\captionsetup{justification=centering}
		\,\,\,\,
		\begin{tabular}{ l } 
			$	 \vop{x} \ta  \ta  \vcl{x} \vop{y} \tb \tb \vcl{y} \tb $ \\ 
			$	 \ta \ta \vop{x} \ta  \ta  \vcl{x} \vop{y} \tb \tb \vcl{y} \tb $   \\
			$	 \ta \ta \vop{x} \ta   \vcl{x} \vop{y}  \tb \vcl{y} \tb $   \\
		\end{tabular}
	\end{minipage}
	\end{center}
 \caption{\label{fig:introc} Production rules and some of the ref-words they produce.}
\end{figure}
For example, the extraction grammar on the left of Figure~\ref{fig:introc} produces also the ref-words $\vop{x} \ta \vcl{x} \vop{y} \tb \vcl{y} \ta \tb \tb $ and $\ta \tb \vop{x} \ta \vcl{x} \vop{y} \tb \vcl{y} \tb $. 
Hence, it extracts from $\doc \df \ta \tb \ta \tb \tb$ the two first tuples from the relation in Figure~\ref{fig:introb}.
In the right Figure~\ref{fig:introc}, there are additional examples of ref-words produced by this grammar.
In general, the given grammar extracts from documents the spans $x$ and $y$ that correspond, respectively, with a non-empty substring of $\ta$'s followed by an equal-length substring of $\tb$'s. 
With a slight adaptation of Fagin et al.~inexpressibility proof \cite[Theorem 4.21]{DBLP:journals/jacm/FaginKRV15}, it can be shown that this spanner is inexpressible by core spanners.

Indeed, we show that context-free spanners are strictly more expressive than regular
spanners, and that the restricted class of regular extraction grammars captures the regular
spanners. We compare the expressiveness of context-free spanners against core and generalized core
spanners and show that context-free spanners are incomparable to any of these classes. 
In addition to extraction grammars, we present a pushdown automata
model that captures the context-free spanners.

In terms of evaluation of context-free spanners, we can evaluate extraction grammars in polynomial time in \emph{data complexity}, where the spanner is regarded as fixed and the document as input. However, as the 
degree of this polynomial depends on the query (in particular, on the number of variables in the relation it extracts),
we propose an enumeration algorithm for unambiguous extraction grammars. Our algorithm outputs the results consecutively, after quintic preprocessing, with constant delay between every two answers. 
In the first step of the preprocessing stage, we manipulate the extraction grammar so that it will be adjusted to the input document. Then, in the second step of the preprocessing, we change it in a way that its non-terminals include extra information on the variable operations. This extra information enables us to skip sequences of productions that do not affect the output, hence obtaining a delay that is independent of the input document, and linear in the number of variables associated with the spanner.   

\paragraph*{Related Work}
Grammar-based parsers are widely used in IE systems~\cite{yakushiji2000event, seoud2007extraction}. There are, as well, several theoretical frameworks that use grammars for IE, one of which
is Knuth's framework of attribute grammars~\cite{DBLP:journals/mst/Knuth68, DBLP:journals/mst/Knuth71}.
In this framework, the non-terminals of a grammar are attached with attributes\footnote{The term ``attributes'' was previously used in the relational context; 
	Here the meaning is different.} that pass semantic information up and down a parse-tree. 
While both extraction grammars and attribute grammars extract information via grammars, it seems as if their expressiveness is incomparable.

The problem of enumerating words of context-free grammars arises in different contexts~\cite{wen2004enumerating,nagashima1987formal}. 
Providing complexity guarantees on the enumeration is usually tricky and requires assumptions either on the grammar or on the output. M\"{a}kinen~\cite{makinen1997lexicographic} has presented an enumeration algorithm for regular grammars and for unambiguous context-free grammars with additional restrictions (strongly prefix-free and length complete). Later, D\"{o}m\"{o}si~\cite{domosi2000unusual} has presented an enumeration algorithm for unambiguous context-free grammars that outputs, with quadratic delay, only  words of a fixed length. 

Very recently, while this paper was under review, an extension of it was presented~\cite{amarilli2022efficient}. We discuss it in Section~\ref{sec:conc}. 

\paragraph*{Organization}
In Section~\ref{sec:cfs}, we present  extraction grammars and extraction pushdown automata. In Section~\ref{sec:exp}, we discuss the expressive power of context-free spanners and their evaluation. In Sections~\ref{sec:enum} and~\ref{sec:en}, we present our enumeration algorithm, and in Section~\ref{sec:conc}
we conclude.

\section{Context-Free Spanners}\label{sec:cfs}

In this section we present the class of context-free spanners by presenting two formalisms for expressing them: extraction grammars and extraction pushdown automata.

\subsection{Preliminaries}

We start by presenting the formal setup based on notations and definitions used in previous works on document spanners (e.g.,~\cite{DBLP:journals/jacm/FaginKRV15, DBLP:conf/pods/FreydenbergerKP18}).

\paragraph*{Strings and Spans}
We set an infinite set $\vars$ of variables, and fix a finite alphabet $\alphabet$ that is disjoint of $\vars$. 
In what follows we assume that our alphabet $\Sigma$ consists of at least two letters. 
A \emph{document} $\doc$ is a finite sequence over $\alphabet$ whose length is denoted by $|\doc|$.
A \emph{span} identifies a substring of $\doc$ by specifying its bounding indices. 
Formally, if $\doc = \sigma_1 \cdots \sigma_n$ where $\sigma_i \in \alphabet$ then a span of $\doc$ has the form $\mspan{i}{j}$ where $1 \le i\le j \le n+1$ and $\doc_{\mspan{i}{j}}$ denotes the substring $\sigma_i \cdots \sigma_{j-1}$. 
When $i=j$ it holds that $\doc_{\mspan{i}{j}}$ equals the empty string, which we denote by $\epsilon$.

\paragraph*{Document Spanners}
Let $X\subseteq \vars$ be a finite set of variables and let $\doc$ be  a document.
An $(X,\doc)$-mapping assigns spans of $\doc$ to variables in $X$.
An $(X,\doc)$-relation is a finite set of $(X,\doc)$-mappings. 
A \emph{document spanner} (or \emph{spanner}, for short)  is a function associated with a finite set $X$ of variables that maps documents $\doc$ into $(X,\doc)$-relations.
A spanner is said to be \emph{Boolean} if it is associated with the empty set.
\footnote{Notice that in this work we define document spanners similarly to~\cite{DBLP:journals/jacm/FaginKRV15} based on $(X,\doc)$-mappings rather than on partial mappings as defined in~\cite{DBLP:conf/pods/MaturanaRV18}.}

\subsection{Extraction Grammars }

The \emph{variable operations} of a variable $x\in \vars$ are $\vop{x}$ and $\vcl{x}$ where, intuitively, $\vop{x}$ denotes the opening of $x$, and $\vcl{x}$ its closing.
For a finite subset $X \subseteq \vars$, we define the set $\Gamma_X \df \{ \vop{x}, \vcl{x} \, \, \vline\, \,  x\in X \}$. That is, $\Gamma_X$ is the set that consists of all the {variable operations} of all variables in $X$. 
We assume that $\Sigma$ and $\Gamma_X$ are
disjoint. 
We extend the classical definition of context-free grammars~\cite{DBLP:books/daglib/0016921} by treating the  variable operations as special terminal symbols. 
Formally, a \emph{context-free extraction grammar}, or \emph{extraction grammar} for short, is a tuple 
$G \df ( X, \nter, \ter ,\prodrules, \initvar)$ where
\begin{itemize}
	\item 
	$X \subseteq \vars$ is a finite set of variables,
	\item $\nter$ is a finite set of \emph{non-terminal} symbols\footnote{Note that these are often referred to as variables, however, here we use the term `non-terminals' to distinguish between these symbols and the elements of $\vars$.}, 
	\item $\ter $ is a finite set of \emph{terminal} symbols;
	\item $\prodrules $ is a finite set of \emph{production rules} of the form 
	$A \rightarrow \alpha$ where $A$ is a non-terminal and  $\alpha \in(\nter \cup \Sigma \cup  \Gamma_X)^*$, and
	\item
	$S$ is a designated non-terminal symbol referred to as the \emph{start symbol}.
\end{itemize}
We say that the extraction grammar $G$ is \emph{associated} with $X$.

\begin{example}~\label{ex:glen}
	In this and in the following examples we often denote the elements in $\nter$ 
	by upper 
	case alphabet letters from the beginning of the Latin alphabet ($A,B,C,\ldots$).
	Let $\Sigma = \{ \ta, \tb \}$, and let us consider the grammar
	$\grmr{disjEqLen}$ associated with the variables $\{x,y\}$ that is given by the following production rules:
	\begin{itemize}
		\item 
		$S \rightarrow B  \vop{x}   A  \vcl{y} B \, \vl \, B  \vop{y}   A  \vcl{x} B $
		\item
		$A \rightarrow \ta  A  \ta \, \vl \,  \ta  A  \tb \, \vl \,  \tb  A  \tb \, \vl \,  \tb  A  \ta$
		\item
		$A \rightarrow\,  \vcl{x} B \vop{y} \, \vl \,  \vcl{y} B \vop{x}$
		\item
		$B \rightarrow \epsilon \, \vl \, \ta B \, \vl \, \tb B $
	\end{itemize}
	Here and in what follows, we use the compact notation for production rules by writing $A\rightarrow \alpha_1 \, \vl \, \cdots \, \vl \, \alpha_n$ instead of the productions $A\rightarrow \alpha_1, \cdots, A\rightarrow \alpha_n$.
	As we shall later see, this grammar extracts pairs of disjoint spans with the same length.
\end{example}

Context-free grammars generate words over a given alphabet $\Sigma$, whereas extraction grammars generate words over the extended alphabet $\Sigma \cup \Gamma_X$. These words are referred to as \emph{ref-words}~\cite{sch:cha}.
Similarly to (classical) context-free grammars, the process of deriving ref-words is defined via the notations  $\prodrel,\prodrel^n,\prodrel^*$ that stand for one, $n$, and several (possibly zero) derivation steps, respectively. To emphasize the grammar being discussed, we sometimes use the grammar as a subscript (e.g., $\prodrel^*_G$). 
For the full  definitions we refer the reader to Hopcroft et al.~\cite{DBLP:books/daglib/0016921}. 
A non-terminal $A$ is called \emph{useful} if there is some derivation of the form $S\Rightarrow^* \alpha A \beta \Rightarrow^* w$ where $w\in (\ter\cup \Gamma_X )^*$. 
If $A$ is not useful then it is called \emph{useless}.
For complexity analysis, we define the \emph{size} $|G|$ of an extraction grammar $G$ as the sum of the number of symbols at the right-hand sides (i.e., to the right of $\rightarrow$) of its rules. 
We say that a derivation $A\Rightarrow^* \gamma$ is \emph{terminal} if $\gamma$ is a sequence of terminals. 

\subsection{Semantics of Extraction Grammars }\label{sec:refwords}

Following Freydenberger~\cite{DBLP:journals/mst/Freydenberger19} we define the semantics of extraction grammars using {ref-words}.
A {ref-word} 
$\br\in(\Sigma\cup \Gamma_X)^*$ is 
\emph{valid (for $X$)} if each variable of 
$X$ is opened and then closed
exactly once, or more formally, for each $x\in X$ the string
$\br$ has precisely one occurrence of $\vop{x}$,  precisely one
occurrence of $\vcl{x}$, and the former is before
(i.e., to the left of) the latter.

\begin{example}\label{ex:refval}
	The ref-word 
	$\br_1 \df \, \vop{x} \ta \ta \vcl{y} \vop {x} \ta \tb \vcl{y}$ is not valid for $\{x,y\}$ whereas the ref-words
	$\br_2 \df \, \vop{x} \ta \ta \vcl{x} \vop {y} \ta \tb \vcl{y}$, 	$\br_3 \df \, \vop{y} \ta \vcl{y} \vop{x} \ta \vcl{x} \ta \tb $ 
 and $\br_4 \df \, \vop{y} \ta \vop{x}  \vcl{y} \ta \vcl{x} \ta \tb $ are valid for $\{x,y\}$.  
\end{example}

To connect ref-words to terminal strings and later to
spanners, we define a morphism 
$\clean\colon (\Sigma\cup \Gamma_X)\wild\to \Sigma\wild$ by 
$\clean(\sigma)\df\sigma$ for 
$\sigma\in\Sigma$, and 
$\clean(\tau)\df \epsilon$
for $\tau\in\Gamma_X$. For $\doc\in\Sigma\wild$, let
$\refl(\doc)$ be the set of all valid ref-words
$\br
\in(\alphabet\cup\Gamma_X)\wild$ 
with 
$\clean(\br)=\doc$.  
By
definition, every $\br\in \refl(\doc)$ 
has a unique factorization
$\br = \br_x'\cdoto 
\vop{x}\cdoto \br_{x}\cdoto \vcl{x}\cdoto
\br_x''$ 
for each $x\in X$. With these factorizations, we interpret
$\br$ as a $(X,\doc)$-mapping $\mu^{\br}$ by defining
$\mu^\br(x) \df \mspan{i}{j}$, 
where $i\df |\clean(\br'_x)|+1$ and
$j\df i+|\clean(\br_{x})|$.  
An alternative way of understanding $\mu^{\br}(x)=\mspan{i}{j}$ is that
$i$ is chosen such that $\vop{x}$ occurs between the positions in
$\br$ that are mapped to $\sigma_{i-1}$ and $\sigma_{i}$, and
$\vcl{x}$ occurs between the positions that are mapped to
$\sigma_{j-1}$ and $\sigma_{j}$ (assuming that $\doc = \sigma_1
\cdots \sigma_{|\doc|}$, and slightly abusing the notation to avoid a
special distinction for the non-existing positions $\sigma_{0}$ and
$\sigma_{|\doc|+1}$).
\begin{example}
	Let $\doc = \ta \ta \ta \tb$.
	The ref-word $\br_2$ from Example~\ref{ex:refval} is interpreted as the $(\{x,y \}, \doc)$-mapping $\mu^{\br_2}$ defined by 
	$\mu^{\br_2}(x)\df\mspan{1}{3}$ and $\mu^{\br_2}(y)\df\mspan{3}{5}$.
 The ref-word $\br_3$ is interpreted as $ \mu^{ \br_3} $ defined by  $ \mu^{\br_3}(x)\df\mspan{2}{3} $ and $\mu^{\br_3}(y)\df\mspan{1}{2} $, and $\br_4$ as $ \mu^{ \br_3} $ as well.
\end{example}

Extraction grammars define ref-languages which are sets of ref-words. 
The ref-language $\rlang(G)$ of an extraction grammar 
$G \df (X, \nter, \ter, \prodrules, \initvar) $ is defined by
$
\rlang(G) \df
\{ \br \in (\ter \cup \Gamma_{X} )^*
\, \vline \,\,
\initvar  \prodrel^* \br \}\,.
$
Note that we use $\rlang {(G)}$ instead of  $\mathcal L{(G)}$ usually used to denote the language of standard grammars, to emphasize that the produced language is  a ref-language. (We also use $\mathcal L{(G)}$ when $G$ is a standard grammar.)
To illustrate the definition let us consider the following example.
\begin{example}\label{ex:produce}
	\sloppy{Following  Examples~\ref{ex:glen} and \ref{ex:refval}, notice that the ref-words $\br_1,\br_2$ and $\br_3$  are in $\rlang(\grmr{disjEqLen})$ whereas $\br_4$ is not.
	Producing both $\br_1$ and $\br_2$ starts similarly with the sequence:
	$
	\initvar \prodrel\, B\vop{x}A\vcl{y}B \, \prodrel^2\, \vop{x}   A  \vcl{y}  \,\prodrel \, \vop{x} \ta  A \tb \vcl{y} \, \prodrel \, \vop{x} \ta  \ta A \ta \tb \vcl{y}$.
	The derivation of $\br_1$ continues with
	$
	\prodrel \,
	\vop{x} \ta  \ta  \vcl{y} B \vop{x} \ta \tb \vcl{y}\,
	\prodrel \, \vop{x} \ta  \ta  \vcl{y} \vop{x} \ta \tb \vcl{y}
	$
	whereas that of $\br_2$  continues with
	$
	\prodrel\,
	\vop{x} \ta  \ta  \vcl{x} B \vop{y} \ta \tb \vcl{y}\,
	\prodrel \,\vop{x} \ta  \ta  \vcl{x} \vop{y} \ta \tb \vcl{y}
	$.	 }
\end{example}
We denote by $\refl(G)$ the set of all
ref-words in 
$\rlang(G)$ that are valid for $X$.
Finally, we define  
the set $\refl(G,\doc)$ of ref-words in $\refl(G)$ that $\clean$ maps to $\doc$. That is, 
$\refl(G,\doc) \df  \refl(G) \cap \refl(\doc)\,.$  
The result of evaluating the spanner $\repspnr{G}$ on a document $\doc$ is then defined as 
\[
\repspnr{G}(\doc) \df \{ \mu^{\br} \, \, \vline \, \, \br \in \refl(G,\doc) \}\,.
\]
\begin{example}
Following the previous examples in this section, 
	let us consider the document $\doc \df \ta \ta \tb \ta$.
	The grammar $\grmr{disjEqLen}$ maps $\doc$
	into a set of $(\{x,y \}, \doc)$-mappings, amongst are 
	$ \mu^{\br_2} $ that is defined by  $ \mu^{\br_2}(x) \df\mspan{1}{3} $ and $\mu^{\br_2}(y)\df\mspan{3}{5} $, and 
	$ \mu^{ \br_3} $ that is defined by  $ \mu^{\br_3}(x)\df\mspan{2}{3} $ and $\mu^{\br_3}(y)\df\mspan{1}{2} $. 
 It can be shown that 
	the grammar $\grmr{disjEqLen}$ maps every document $\doc $ into all possible $(\{x,y \},\doc)$-mappings $\mu$ such that
	$\mu(x)$ and $\mu(y)$ are disjoint (i.e., do not overlap) and have the same length (i.e.,
	$|\doc_{\mu(x)}| =  |\doc_{\mu (y)}|$).
\end{example}
A spanner $S$ is said to be \emph{definable} by an extraction grammar $G$ if $S(\doc) = \repspnr{G}(\doc)$ for every document $\doc$. 
\begin{definition}
	A \emph{context-free spanner} is a spanner definable by an extraction grammar. 
\end{definition}

\subsection{Extraction Pushdown Automata}\label{sec:epa}
An \emph{extraction pushdown automaton}, or \emph{extraction $\PDA$}, is associated with a finite set $X\subseteq \vars$ of variables and can be viewed as a standard pushdown automata over the extended alphabet $\Sigma \cup \Gamma_X$.
Formally,
an  \emph{extraction $\PDA$} is a tuple 
$A \df (X, Q,\Sigma, \Delta, \delta, q_0, Z,F)$ where
$X$ is a finite set of variables;
$Q$ is a finite set of states;
$\Sigma$ is the input alphabet;
$\Delta$ is a finite set which is called \emph{the stack alphabet};
$\delta$ is a mapping 
$Q \times \big( \Sigma \cup \{ \epsilon\} \cup \Gamma_X \big) \times \Delta \rightarrow 2^{Q \times \Delta^*}$ which is called the \emph{transition function};
$q_0 \in Q$ is the \emph{initial state};
$Z \in \Delta$ is the \emph{initial stack symbol}; and
$F \subseteq Q$ is the set of \emph{accepting states}.
Indeed, 
extraction $\PDA$s run on ref-words (i.e., finite sequences over $\Sigma \cup \Gamma_X$), as opposed to classical $\PDA$s whose input are words (i.e., finite sequences over $\Sigma$). 
Similarly to classical $\PDA$s, the computation of extraction $\PDA$s can be described using sequences of configurations: a \emph{configuration} of $A$ is a triple $(q,w,\gamma)$ where $q$ is the state,
$w$ is the remaining input, and
$\gamma$ is the stack content such that
the top of the stack is the left end of $\gamma$ and its bottom is the right end.
We use the notation $
{\moverel}^*$ similarly to how it is used in the context of $\PDA$s~\cite{hop:int}
and define the ref-language $\mathcal{R}{(A)}$:
\[
\rlang{(A)} \df
\{ \br \in (\ter \cup \Gamma_{X} )^*
\, \, \vline \,\, \exists \alpha \in \Delta^*,q_f \in F:
(q_0, \br, Z) \moverel^* (q_f, \epsilon, \alpha) \}\,. 
\]
We denote the language of $A$ by $\mathcal R(A)$ to emphasize that it is a ref-language, and denote by $\refl(A)$ the set of all
ref-words in 
$\rlang(A)$ that are valid for $X$. 
The result of evaluating the spanner $\repspnr{A}$ on a document $\doc$ is then defined as 
\[
\repspnr{A}(\doc) \df \{ \mu^{\br} \, \, \vline \, \, \br \in  \refl(A) \cap \refl(\doc) \}\,.
\]

\begin{example}\label{ex:pda}
	We define the extraction $\PDA$ that maps a document $\doc$ into the set of $(\{x,y\}, \doc)$-mappings $\mu$ where $\mu(x)$ ends before $\mu(y)$ starts and their lengths are the same. The stack alphabet consists of the bottom symbol $\bot$ and $C$, and
	the transition function
	$\delta$ is described in Figure~\ref{fig:pda} where
	a transition from state $q$ to state $q'$ that is labeled with  $\tau, A/\ \gamma $ denotes that the automaton moves from state $q$ to state $q'$ upon reading $\tau$ with $A$ at the top of the stack, while replacing $A$ with $\gamma$.
	We can extend the automaton in a symmetric way such that it will represent the same spanner as that represented by the grammar
	$\grmr{disjEqLen}$ from Example~\ref{ex:glen}.
\end{example}

\begin{figure}
	\begin{center}
\begin{tikzpicture}[initial text={},->,>=stealth',
	,
	auto,
	node distance=3cm,
transform shape]	
	\node[state,initial] (q_0) {$q_0$};
	\node[state] (q_x) [right of=q_0] {$q_x$};
	\node[state] (q_{xy}) [right of=q_x] {$q_{xy}$};
	\node[state] (q_y) [right of=q_{xy}] {$q_y$};
	\node[state,accepting] (q_f) [right of=q_y] {$q_f$};
	
	\path (q_0) edge              node {$\scriptsize{\vop{x}, C/\ C}$} (q_x)
	(q_x) edge              node {$\scriptsize{\vcl{x}, C/\ C}$} (q_{xy})
	(q_{xy}) edge              node {$\scriptsize{\vop{y}, C/\ C}$} (q_y)
	(q_y) edge              node {$\scriptsize{\vcl{y}, \bot /\ \bot}$} (q_f)
	(q_0) edge      [loop above]        node {$\scriptsize{\Sigma, \bot/\  \bot}$} (q_0)
	(q_f) edge      [loop above]        node {$\scriptsize{\Sigma, \bot/\  \bot}$} (q_f)
	(q_x) edge [loop above]              node {$\scriptsize{\Sigma, C /\ C C}$} (q_x)
	(q_{xy}) edge   [loop above]           node {$\scriptsize{\Sigma, C /\ C}$} (q_{xy})
	(q_y) edge      [loop above]        node {$\scriptsize{\Sigma, C /\ \epsilon}$} (q_y);
\end{tikzpicture}
\caption{\label{fig:pda} Transition function of Example~\ref{ex:pda}}
\end{center}
\end{figure}

We say that a spanner $S$ is \emph{definable}  
by an extraction $\PDA$ $A$ if for every document $\doc$ it holds that $\repspnr{A}(\doc) = S(\doc)$. 
Treating the variable operations as terminal symbols enables us to use the equivalence of $\PDA$s and context-free grammars and conclude the following straightforward observation.
\def\propepda{The class of spanners definable by extraction grammars is equal to the class of spanners definable by extraction PDAs.}
\begin{proposition}\label{prop:epda}
	\propepda
\end{proposition}
Thus, there is also an automata formalism for defining context-free spanners. 
We now present important classes of extraction grammars for which we later present our evaluation algorithm. 

\subsection{Functional Extraction Grammars}

Freydenberger and Holldack~\cite{DBLP:conf/icdt/FreydenbergerH16}
have presented the notion of \emph{ functionality} in the context of regular spanners. We now extend it to extraction grammars. 
The intuition is that interpreting an extraction grammar as a spanner disregards  ref-words that are not valid.
We call an extraction grammar $G$ \emph{functional} if every ref-word in $\mathcal{R}(G)$ is valid.
\begin{example} \label{ex:lentofunc}
	The grammar $\grmr{disjEqLen}$ in our running example is not functional. Indeed, we saw in Example~\ref{ex:produce} that the ref-word $\br_1$, although it is not valid, is in $\mathcal{R}(\grmr{disjEqLen})$.
	We can, however, simply modify the grammar to obtain an equivalent functional one. 
	Notice that the problem arises due to the production rules 
	$
	S \rightarrow B  \vop{x}   A  \vcl{y} B $ and
	$S \rightarrow B  \vop{y}   A  \vcl{x} B$.
	For the non-terminal $A$ we have $A\Rightarrow^* \br_1$ where $\br_1$ 
	contains both $\vcl{x}$ and $\vop{y}$, and we also have $A\Rightarrow^* \br_2$ where $\br_2$ 
	contains both  $\vcl{y}$ and $\vop{x}$.
	To fix that, 
	we can replace the non-terminal $A$ with two non-terminals, namely $A_1$ and $A_2$, and change the production rules so that for every ref-word $\br$, if $A_1\Rightarrow^* \br$ then  $\br$ 
	contains both $\vcl{x}$ and $\vop{y}$, and if 
	$A_2\Rightarrow^* \br$ then  $\br$ 
	contains both $\vcl{y}$ and $\vop{x}$.
	It can be shown that  the grammar $G$ whose production rules appear in Figure~\ref{fig:rules} is   functional and that $\repspnr{G} = \repspnr{\grmr{disjEqLen}}$.
\end{example}

\begin{figure}
	\centering
	\begin{tabular}[scale=1]{l}
		$\scriptsize{S \rightarrow B  \vop{x}   A_1  \vcl{y} B\, \vl\,  B  \vop{y}   A_2  \vcl{x} B} $ \\
		$\scriptsize{A_i \rightarrow \ta  A_i  \ta \,\vl  \ta A_i  \tb\, \vl \, \tb  A_i  \tb \,\vl \, \tb  A_i  \ta,\,\, i=1,2}$
		\\
		$\scriptsize{A_1 \rightarrow\, \vcl{x} B \vop{y}}$ , \, 
		$\scriptsize{A_2 \rightarrow \, \vcl{y} B \vop{x}}$
		\\
		$\scriptsize{B \rightarrow \epsilon\, \vl\, \ta B\, \vl\, \tb B} $
	\end{tabular}
	\hfill
	\caption{\label{fig:rules}Productions  of Example~\ref{ex:lentofunc}}
\end{figure}

Recall that a context-free grammar is said to be in \emph{Chomsky Normal Form (CNF)} if all of its production rules are of the form $A\rightarrow BC$ or $A\rightarrow \ta$ where $A,B,C$ are non-terminals and $\ta$ is a terminal.
We extend this notion to extraction grammars. We say that an extraction grammar is in CNF if it is in CNF when viewed as a grammar over the extended alphabet $\Sigma \cup \Gamma_X$. 

\def\propfunc{Every extraction grammar $G$  can be converted into an equivalent functional extraction grammar $G'$
	in  $O(|G|^2+3^{2k}|G|^2)$ time where $k$ is the number of variables $G$ is associated with. In addition, $G'$ is in CNF.
}
\begin{proposition}~\label{prop:func}
	\propfunc
\end{proposition}

\begin{proof}
	Let $G \df ( X, \nter, \ter ,\prodrules, \initvar)$ be an extraction grammar. 
	We start by converting $G$ to CNF with the standard algorithm presented, e.g., in~\cite[Section 7.1.5]{hop:int}:
		The algorithm consists of three steps:
	\begin{enumerate}
		\item  We omit $\epsilon$-productions, unit productions and useless symbols. 
		\item
		We arrange that all right-hand sides of length 2 or more consists only of variables by adding productions of the form $A\rightarrow \sigma$;
		\item
		We break right-hand sides of length 3 or more into a cascade of productions, each with a right-hand side consisting of two variables. This is done by replacing rules of the form $A \rightarrow A_1 \cdots A_n$ into the cascade $A\rightarrow A_1 B_1,\, B_1\rightarrow A_2 B_2,\,\cdots\,B_{n-2} \rightarrow A_{n-1} A_n$.
	\end{enumerate}
	
	By a slight abuse of notation, we refer to the resulting grammar by the same notation.

	We next define an extraction grammar $G'$ whose non-terminals indicate which variable operations they generate. 
	Formally, $G'$ is associated with $X$, its non-terminals are $X \times 2^{\Gamma_X}$, its terminals are $\ter$, its start symbol is $(S,\Gamma_X)$, and its production rules are the following:
	\begin{itemize}
		\item 
		$(A,X_1)\rightarrow (B,X_2)(C,X_3)$ whenever \begin{itemize}
			\item $A\rightarrow BC \in \prodrules$,
			\item  $X_1  = X_2 \cup X_3$, and
			\item for every $x\in X$ if $\vcl{x} \in X_2$ then $\vop{x}\not\in X_3$.
		\end{itemize} 
		\item 	$(A,\emptyset)\rightarrow \sigma$ whenever  $A\rightarrow \sigma \in \prodrules$ with $\sigma \in \Sigma$; and
		\item
		$(A,\{\tau\})\rightarrow \tau$ whenever  $A\rightarrow \tau \in \prodrules$ with $\tau \in \Gamma_X$.
	\end{itemize}  
We note that $G'$ is in CNF, and show now that 
 $G'$ is both functional and equivalent to $G$.
	
	\paragraph{$G'$ is functional} A straightforward induction shows that whenever 
	$(S,\Gamma_X) \Rightarrow^* \alpha_1 \cdots \alpha_n$ with $\alpha_i$ either a pair $(A_i,X_i)$ or a terminal $\beta_i$, it holds that whenever $\vcl{x} \in X_i$, for every $j>i$ it holds that $\vop{x} \not \in X_j$ or $\beta_j \ne \vop{x}$. This allows us to conclude that in a terminal derivation $(S,\Gamma_X) \Rightarrow^* \alpha_1 \cdots \alpha_n$ of $G'$,
	it holds that $\alpha_1 \cdots \alpha_n$ is a valid ref-word.
	
	\paragraph{$G'$ is equivalent to $G$}
 Again, by induction on the length of production, we show that if $(S,\Gamma_X) \Rightarrow^*_{G'} \alpha_1 \cdots \alpha_n$ with $\alpha_i$  either a pair $(A_i,X_i)$ or  $\beta_i\in \Sigma$ then $S \Rightarrow^*_{G} \alpha'_1 \cdots \alpha'_n$ with $\alpha'_i = A_i$ whenever $\alpha_i\df(A_i,X_i)$, and $\alpha'_i = \beta_i$ whenever  $\alpha_i\df\beta_i$.
	 Thus, if $(S,\Gamma_X) \Rightarrow^*_{G'} \beta_1 \cdots \beta_n$ is a derivation for which $\beta_1 \cdots \beta_n \in (\Sigma \cup \Gamma_X)^*$ then so is $S \Rightarrow^*_{G} \beta_1 \cdots \beta_n$.
	 For the other direction, a straightforward induction shows that if  $A \Rightarrow^*_{G} \alpha'_1 \cdots \alpha'_n$ then $(A,Y) \Rightarrow^*_{G'} \beta_1 \cdots \beta_n$ 
	 where $\beta_i = \alpha'_i$ whenever $\alpha'_i\in \ter$, and $\beta_i = (\alpha'_i,Y_i)$ whenever $\alpha'_i \in \nter$ for some $Y_i \subseteq \Gamma_X$.
Thus, if  $S \Rightarrow^*_{G} \alpha'_1 \cdots \alpha'_n$ is a  derivation for which $\alpha'_1 \cdots \alpha'_n \in (\Sigma \cup \Gamma_X)^*$ then so is $(S,\Gamma_{X}) \Rightarrow^*_{G'}  \alpha'_1 \cdots \alpha'_n$.

	 \paragraph{Complexity}	
 For the complexity analysis we note that converting $G$ into CNF requires $O(|G|^2)$ and that the resulting grammar is of size $O(|G|^2)$. Since the
 number of ways to choose two disjoint subsets of $\Gamma_X$ is $3^{2k}$,
 the construction requires $O(|G|^2+ 3^{2k}|G|^2)$ where $G$ is associated with $k$ variables.
\end{proof}
We remark that if $G$ is in CNF, the complexity of converting it to $G'$ reduces to $O(3^{2k}|G|)$.

\subsection{Unambiguous Extraction Grammars}

A grammar $G$ is said to be unambiguous if every word it produces has a unique parse-tree.
We extend this definition to extraction grammars as follows. 
An extraction grammar $G$ is said to be \emph{unambiguous} if for every document $\doc$ and every $(X,\doc)$-mapping $\mu \in \repspnr{G}(\doc)$ it holds that
\begin{itemize}
	\item  there is a unique ref-word $\br$ for which $\mu^{\br} = \mu$, and 
	\item $\br$ has a unique parse-tree.
\end{itemize}
Unambiguous extraction grammars are less expressive than their ambiguous counterparts. Indeed, context-free grammars (which are extraction grammars associated with the empty set of variables) are less expressive than ambiguous context-free grammars~\cite{DBLP:books/daglib/0016921}.
\begin{example}
	The extraction grammar given in Example~\ref{ex:lentofunc} is not unambiguous since 
	it produces the ref-words $ \vop{x}\vcl{x} \vop{y} \vcl{y}$ and $\vop{y} \vcl{y} \vop{x} \vcl{x}$ that correspond to the same mapping.
	It can be shown that replacing the derivation $B\rightarrow \epsilon$ with $B\rightarrow \ta \vl \tb$ results in an unambiguous extraction grammar which is equivalent to $\grmr{disjEqLen}$ on any document different than $\epsilon$. 
	(Note however that this does not imply that the ref-languages both grammars produce are equal.)
\end{example}
Our enumeration algorithm for extraction grammars presented in Section~\ref{sec:en}
relies on unambiguity and the following observation.
\def\prepunambig{In Proposition~\ref{prop:func}, if $G$ is unambiguous then so is $G'$.}
\begin{proposition}\label{prop:unambig}
	\prepunambig
\end{proposition}
\begin{proof}
	Assume  that $G$ is unambiguous. 
	
	We first show that converting $G$ to CNF results in an 
	unambiguous extraction grammar.
	For that, let us recall the details of the algorithm presented in~\cite[Section 7.1.5]{hop:int} for converting a grammar to CNF.
	The algorithm consists of three steps:
	\begin{enumerate}
		\item \label{i1} We omit $\epsilon$-productions, unit productions and useless symbols. 
		\item \label{i2}
		We arrange that all bodies of length 2 or more consists only of variables;
		\item \label{i3}
		We break bodies of length 3 or more into a cascade of productions, each with a body consisting of two variables.
	\end{enumerate}
	In step \ref{i1}, we omit productions and thus if $G$ is unambiguous then so is the resulting grammar. 
	In step \ref{i2}, we add productions of the form $A\rightarrow \sigma$ but these do not affect unambiguity as these $A$s are fresh non-terminals and for each $\sigma$ we have a unique such $A$.
	In step \ref{i3}, we replace rules of the form $A\rightarrow A_1 \cdots A_n$ into the cascade $A\rightarrow A_1 B_1, B_1\rightarrow A_2 B_2, \cdots B_{n-2}\rightarrow A_{n-1} A_n$ where all $B_{i}$s are fresh non-terminals. Thus, this step as well does not introduce ambiguity.
	
It is straightforward that the conversion to $G'$ preserves unambiguity, which completes the proof. 
\end{proof}

\section{Expressive Power and Closure Properties}\label{sec:exp}
In this section we compare the expressiveness of context-free spanners compared to other studied classes of spanners and discuss  its evaluation.

\subsection{Regular Spanners}\label{sec:reg}
The most studied language for specifying spanners is that of the regular spanners which are those definable by a nondeterministic
finite-state automaton that can open and close variables while running.
Formally, a \emph{variable-set automaton} $A$ (or \emph{vset-automaton}, for short)  
is 
a tuple
$A\df (X, Q,q_0,q_f,\delta)$ where  $X\subseteq \vars$ is a finite set of variables also referred to as $\vars(A)$, 
$Q$ is the set of \emph{states}, $q_0,q_f\in Q$ are the \emph{initial} and
the \emph{final} states, respectively, and 
$\delta\colon
Q\times(\alphabet\cup\{\epsilon\}\cup\Gamma_X)\to 2^{Q}$ is the
\emph{transition function}. 
The semantics of $A$ is defined by interpreting $A$ as a non-deterministic finite state automaton
over the extended alphabet 
$\alphabet\cup\Gamma_X$,
and defining $\rlang(A)$ as the set of all
ref-words $\br\in (\alphabet\cup\Gamma_X)^*$ such that some path
from $q_0$ to $q_f$ is labeled with $\br$. Formally, 
\[
\rlang{(A)} \df
\{ \br \in (\ter \cup \Gamma_{X} )^*
\, \, \vline \,\, q_f \in \delta(q_0,\br) \}\,. 
\]
Similarly for regex formulas, we define $\refl(A,\doc)=\rlang(A)\cap \refl(\doc)$. Finally, the result of applying the spanner $\repspnr{A}$ on a document $\doc \in \Sigma^*$ is defined as 
\[
\repspnr{A}(\doc) \df \{\mu^{\br} \,\, \vline \,\, \br \in \refl(A,\doc) \}.
\]
The class of \emph{regular spanners} equals the class of spanners that are expressible as a vset-automaton~\cite{DBLP:journals/jacm/FaginKRV15}.

Inspired by Chomsky's hierarchy, we say that an extraction grammar $G$ is \emph{regular} if its productions are of the form $A \rightarrow \sigma B$ and $A \rightarrow \sigma$ where $A,B$ are non-terminals and $\sigma \in (\Sigma \cup \Gamma_X)$. 
We then have the following equivalence that is strongly based on the  equivalence of regular grammars and finite state automata.
\def\thmregspnrs{The class of spanners definable by regular extraction grammars is equal to the class of regular spanners.}
\begin{proposition}\label{thm:regspnrs}
	\thmregspnrs
\end{proposition}
\begin{proof}	
	Let us consider a regular spanner definable by a vset-automaton $A$. Since every regular language is also context-free, we can construct a pushdown automaton $A\p$ such that $\mathcal{R}(A\p) = \mathcal{R}(A)$, which allows us to conclude that  $\repspnr{A} = \repspnr{A\p}$ 
	
	For the other direction we have to show that every spanner definable by a regular extraction grammar is regular. 
	Let $G \df (X,\nter, \ter, \prodrules, \initvar)$ be a regular extraction grammar. 
	We adapt the standard conversion of regular grammars into non-deterministic automata (cf~\cite{DBLP:books/daglib/0016921}) to our settings by treating variable operations as terminal symbols.
	Formally, the automaton is defined as follows. Its set of states consists of $q_A$ for every non-terminal $A \in \nter$, and a fresh final state $q_f$. We set $q_{\initvar}$ as the initial state and define the transition function $\delta$ as follows:
	for every $A \rightarrow \sigma B$ we set $\delta(q_A,\sigma) = q_B$, and for $A\rightarrow \sigma$ we set $\delta(q_A,\sigma) = q_f$.
	It holds that $\mathcal{R}(G) = \mathcal{L}(A)$, which allows us to conclude that $\repspnr{G}(\doc)  = \repspnr{A}(\doc)$ for every document $\doc$.
\end{proof}

\subsection{(Generalized) Core Spanners} 
In their efforts to capture the core of AQL which is IBM's SystemT query language, Fagin et al.~\cite{DBLP:journals/jacm/FaginKRV15} defined the class of core spanners. 
Core spanners
extend regular spanners with the string-equality selection. To define them properly, we present an 
 alternative way of defining regular spanners that is based on the notion of regex formulas:
A \emph{regex formula}
is defined recursively by
\[ \alpha\df
\hspace{3.5pt} \emptyset\mid \epsilon\mid \sigma \mid \alpha \vee
\alpha\mid \alpha \cdot \alpha \mid \alpha^* \mid \, \vop{x}  \alpha \vcl{x}
\]
where $\sigma \in \Sigma$ and $x\in \vars$.
We 
denote the set of variables whose variable operations occur in $\alpha$ by $\vars(\alpha)$,  
and interpret each regex formula $\alpha$ as a generator of a 
ref-word
language $\rlang(\alpha)$ over the extended alphabet
$\alphabet\cup\Gamma_{\vars(\alpha)}$. 
For every
document $\doc \in \alphabet\wild$, we define
$\refl(\alpha,\doc)=\rlang(\alpha)\cap \refl(\doc)$, and
the spanner $\repspnr{\alpha}$ 
by
\[\repspnr{\alpha}(\doc)\df
\{\mu^\br \, \, \vline \, \, \br \in \refl(\alpha,\doc)\}
.\]

The class of regular spanners is then defined as the closure of (spanners defined by) regex formulas under the algebraic operators: union, projection and natural join~\cite{DBLP:journals/jacm/FaginKRV15}.
Let $P$, $P_1$ and $P_2$ be spanners; the above operators are defined as follows:

\paragraph{Union} If $\vars(P_1) = \vars(P_2)$, their \emph{union}
$(P_1 \cup P_2)$ is defined by $\vars(P_1 \cup P_2) \df \vars(P_1)$
and $(P_1 \cup P_2)(\doc) \df P_1(\doc) \cup P_2(\doc)$ for all
$\doc\in\alphabet\wild$.

\paragraph{Projection} Let $Y \subseteq \vars(P)$. The
\emph{projection} $\pi_Y P$ is defined by $\vars(\pi_Y P) \df Y$ and
$\pi_Y P(\doc) \df {P|}_{Y}(\doc)$ for all $\doc \in
\alphabet\wild$, where ${P|}_{Y}(\doc)$ is the restriction of all
$\mu\in P(\doc)$ to $Y$.

\paragraph{Natural join} Let $V_i \df \vars(P_i)$ for $i \in
\{1,2\}$. The \emph{(natural) join} $(P_1 \join P_2)$ of $P_1$ and
$P_2$ is defined by $\vars(P_1 \join P_2) \df \vars(P_1) \cup
\vars(P_2)$ and, for all $\doc \in \alphabet\wild$, $(P_1\join
P_2)(\doc)$ is the set of all $(V_1 \cup V_2, \doc)$-records~$\mu$
for which there exist $\mu_1\in P_1(\doc)$ and $\mu_2\in P_2(\doc)$
with ${\mu|}_{V_1}(\doc) = \mu_1(\doc)$ and ${\mu|}_{V_2}(\doc) =
\mu_2(\doc)$.

The class of core spanners  
is the closure of regex formulas under the positive operators, (i.e., union, natural join and projection) along 
with the string equality selection that is defined as follows:

\paragraph*{String-equality selection} 
Let $P$ be a spanner and let $x,y \in \vars(S)$, the \emph{string equality selection}
$\sel^=_{x,y} S$ is defined by $\vars(\sel^=_{x,y} S) = \vars(S)$ and, for all $\doc \in \Sigma^*$, $\sel^=_{x,y} S(\doc)$ is the set of all $\mu\in
S(\doc)$ where $\doc_{\mu(x)}=\doc_{\mu(y)}$.

Note that unlike the join operator that joins mappings that have identical spans in
their shared variables, the selection operator compares
the substrings of $\doc$ that are described by the spans, and does
not distinguish between different spans that span the same substrings.

The class of \emph{generalized core spanners} is obtained by adding the difference operator. 
\paragraph{Difference} Let $P_1,P_2$ be spanners. If $\vars(P_1)=\vars(P_2)$, their \emph{difference} $P_1 \diff P_2$ is defined by $\vars(P_1 \diff P_2) = \vars(P_1)$ and $(P_1 \diff P_2)(\doc) = P_1(\doc) \diff P_2(\doc)$.

Generalized core spanners are the closure of regex formulas under union, natural join, projection, string equality, and difference.  

We say that two classes $\mathcal{S},\mathcal{S}\p$ of spanners are \emph{incomparable} if both $\mathcal{S} \setminus \mathcal{S}\p$ and $\mathcal{S}\p \setminus \mathcal{S}$ are not empty.
\def\propincomp{The classes of core spanners and generalized core spanners are each incomparable with the class of context-free spanners.}
\begin{proposition}\label{prop:core}
	\propincomp
\end{proposition}
\begin{proof}
	We observe that extraction grammars
	associated with the empty set of variables
	are equivalent to context-free grammars and, hence, capture the context-free languages.

The class of core spanners 
	is not contained in the class of context-free spanners since the language $\{ww \,| \,w\in \Sigma^* \}$ is not context-free yet there is a Boolean core spanner $P$ that defines it~\cite{DBLP:journals/jacm/FaginKRV15} (that is, $P(w') \ne \emptyset$ if and only if $w'\in \{ww \,| \,w\in \Sigma^* \}$).
	This also shows that the class of generalized core spanners 
	is not contained in that of context-free spanners. 
	
	The class of context-free spanners is not contained in that of core spanners since
	the language $\{ \ta^n \tb^n | n \ge 1 \}$ is context free and yet it is not accepted by Boolean core spanners~\cite{DBLP:journals/jacm/FaginKRV15}, and not by generalized core spanners~\cite{DBLP:journals/corr/abs-1912-06110}.
\end{proof}

We conclude by an immediate result on closure properties.
\def\propclsrprp{	The class of context-free spanners is closed under union and projection,
	but is 
	not closed under natural join and difference.
}
\begin{proposition}\label{prop:clsrprp}
	\propclsrprp
\end{proposition}
\begin{proof}
	
	To show closure under projection, 
	let $G \df ( X,\nter, \ter ,\prodrules, \initvar) $ be an extraction grammar, and let $X\p \subseteq X$.
	We define a morphism $\clean_{X\setminus X'}: (\Sigma \cup \nter \cup \Gamma_X) \rightarrow (\Sigma \cup \nter \cup \Gamma_{ X'})^*$ by 
	\[\clean_{X\setminus X'}(\alpha) \df \left\{\begin{matrix}
	 \alpha	& \alpha \in \Sigma \cup \Gamma_{X'}\\ 
	\epsilon	& \alpha \in \Gamma_{X \setminus X'} 
	\end{matrix}\right.\]	
	We then define the grammar $G\p = ( X',\nter, \ter ,\prodrules\p, \initvar)$ with $\prodrules\p $ defined as follows: for every $A\rightarrow \alpha \in \prodrules$ we have $A \rightarrow \clean_{X\setminus X'}(\alpha) \in \prodrules\p$.
	It is straightforward  that $\repspnr{G\p} = \repspnr{\pi_{X\p} G_X}$.
	
	Closure under union can be obtained straightforwardly. 
	Let $G_1,G_2$ be two extraction grammars with $\vars(G_1)=\vars(G_2)$.
	We can construct a grammar $G$ with $\mathcal{R}({G}) = \mathcal{R}(G_1) \cup \mathcal{R}(G_2)$ by defining its production rules as the union of those of $G_1$ and $G_2$ with the addition of productions $S \rightarrow S_1, S\rightarrow S_2$ where $S$ is a fresh non-terminal who is the start symbol of $G$, and $S_1,S_2$ are the start symbols of $G_1,G_2$, respectively.

	Non-closure properties are an immediate consequence of non-closure properties of context-free languages.
\end{proof}

\section{Evaluation and Enumeration}
The \emph{evaluation} problem of extraction grammars is that of computing $\repspnr{G}(\doc)$ where $\doc$ is a document and $G$ is an extraction grammar.
Our first observation is the following.

\def\proppoly{For every extraction grammar $G$ and every document $\doc$ it holds that $\repspnr{G}(\doc)$ can be computed in $O(|G|^2+|\doc|^{2k+3}\,k^3\,k!\,|G|)$ time where $k$ is the number of variables $G$ is associated with.}
\begin{proposition}\label{prop:poly}
	\proppoly
\end{proposition} 
\begin{proof}
	Our evaluation algorithm relies on the celebrated 
	Cocke-Younger-Kasami (CYK) parsing algorithm for context-free grammars~\cite{DBLP:books/daglib/0016921} in CNF, and  
	operates as follows: it iterates over all of the ref-words $\br$ that are $(1)$ valid for $\vars(G)$ and $(2)$ mapped by $\clean$ into $\doc$. For each such ref-word, it uses the CYK algorithm to determine whether it is produced by $G$ by treating $G$ as a standard CFG over the extended alphabet $\Sigma \cup \Gamma_X$, after converting it to CNF.
	
	\paragraph{Complexity}	
	We convert $G$ to CNF in $O(|G|^2)$. There are $O(|d|^{2k}k!)$ valid ref-words, and each is represented by a ref-word of length $O(|\doc|+2k)$.
	For each such ref-word, we use the CYK to check whether it belongs to the language of $G$ in $O((|\doc|+2k)^3|G| )$. 
	Since we repeat the process for every ref-word, we get a total complexity of $O(|G|^2+|\doc|^{2k}k! (|\doc|+2k)^3|G| )$.
\end{proof}

By replacing the CYK algorithm with Valiant's parser~\cite{DBLP:journals/jcss/Valiant75}, we can decrease the complexity to  $O(|G|^2+|\doc|^{2k+\omega}\,k^{\omega}\,k!\,|G|)$ with $\omega <2.373$ the matrix multiplication exponent~\cite{DBLP:conf/stoc/Williams12}.

While the evaluation of context-free spanners can be done in polynomial time in data complexity (where $G$ is regarded as fixed and $\doc$ as input), the output size might be quite big. To be more precise, for an extraction grammar $G$ associated with $k$ variables, the output may consist of up to $O({(2k)}!|\doc|^{2k})$ mappings.
Instead of outputting these mappings altogether, we can output them sequentially (without repetitions) after some preprocessing. 
This approach leads us to the main result of this paper:
\def \thmenum{For every unambiguous extraction grammar $G$ and every document $\doc$ there is an algorithm that outputs the mappings in $\repspnr{G}(\doc)$ with delay $O(k)$
	after $O(|\doc|^5|G|^23^{4k})$ preprocessing where $k$ is the number of variables $G$ is associated with.}
\begin{theorem}~\label{thm:enum}
	\thmenum
\end{theorem}

In the rest of this paper, we describe the algorithm, proof its correctness, and analyze its complexity.

Our algorithm consists of two main stages: preprocessing and outputting. 
In the preprocessing stage, we manipulate the extraction grammar and do some precomputations which are later exploited in the outputting stage in which we output the results sequentially. 
We remark that unambiguity is crucial for outputting the mappings without repetition.

Through the lens of data complexity, our enumeration algorithm outputs the results with constant delay after quintic preprocessing. That should be contrasted with regular spanners for which there exists a constant delay enumeration algorithm whose preprocessing is linear~\cite{DBLP:conf/icdt/AmarilliBMN19,DBLP:conf/pods/FlorenzanoRUVV18}.

In the following sections, we present the enumeration algorithm and discuss its correctness but, before, 
we deal with the special case $\doc \df
\epsilon$. In this case, $\repspnr{G}(\doc)$ is either empty or contains exactly one mapping (since, by definition, the document $\epsilon$ has exactly one span, namely $\mspan{1}{1}$).
Notice that  $\repspnr{G}(\doc)$ is empty if and only if $G$ does not produce a ref-word that consists only of variable operations. To check if that is the case, it suffices to change the production rules of $G$ by replacing every occurrence of $\tau\in \Gamma_{X}$ with $\epsilon$, and to check whether the new grammar produces $\epsilon$. This can be done in linear time~\cite{hop:int}, which completes the proof of this case.
So from now on it is assumed that $\doc \ne \epsilon$.

\section{Preprocessing Stage of the Enumeration Algorithm}\label{sec:enum}

The preprocessing stage consists of two steps: in the first we adjust the extraction grammar to a given document and add subscripts to non-terminals to track this connection, and in the second we use superscripts to capture extra information regarding the variable operations.
In addition, we compute a function, namely $\jmp$, that allows us to obtain the complexity guarantees on the delay between two consecutive outputs in the output stage.
\subsection{Adjusting the Extraction Grammar to $\doc$}\label{sec:adj}
Let  ${G}\df( X, \nter, \ter ,\prodrules, \initvar)$ be an extraction grammar in CNF, and
let $\doc \df \sigma_1, \cdots, \sigma_n, n\ge 1$ be a document.
The goal of this step is to restrict $G$ so that it will produce only the ref-words which $\clean$ maps to $\doc$. 
To this end, we define the grammar $G_{\doc}$ that is associated with the same set $X$ of variables as $G$, and is defined as follows:
\begin{itemize}
	\item The non-terminals are $\{A_{i,j} \vl A \in \nter,\, 1\le i \le j \le n\} \cup \{ A_{\epsilon} \vl A \in \nter \} $,
	\item the terminals are $\Sigma$,
	\item the initial symbol is $S_{1,n}$, and
	\item the production rules are defined as follows:
	\begin{itemize}
		\item
		$A_{i,i}\rightarrow \sigma_{i}$ for any $A\rightarrow \sigma_i \in \prodrules$,
		\item
		$A_{\epsilon} \rightarrow \gamma$ for any  $A\rightarrow \gamma \in \prodrules$ with $\gamma \in \Gamma_X$,
		\item 
		$A_{\epsilon} \rightarrow  B_{\epsilon} C_{\epsilon}$
		for any
		$A \rightarrow B C  \in \prodrules$,
		\item 
		$A_{i,j} \rightarrow  B_{i,j} C_{\epsilon}$
		for any $1\le i \le  j \le n$ and any
		$A \rightarrow B C  \in \prodrules$,
		\item
		$A_{i,j} \rightarrow  B_{\epsilon} C_{i,j}$
		for any $1\le i \le j \le n$ and any 
		$A \rightarrow B C  \in \prodrules$,
		\item 
		$A_{i,j} \rightarrow  B_{i,i'} C_{i'+1,j}$
		for any $1\le i \le i' < j \le n$ and 
		$A \rightarrow B C  \in \prodrules$.
	\end{itemize} 
\end{itemize} 
We  eliminate useless non-terminals from $G_{\doc}$, and by a slight abuse of notation refer to the resulting grammar as $G_{\doc}$. The intuition behind this construction is that if a  subscript of a non-terminal is $i,j$ then this non-terminal produces a ref-word that $\clean$ maps to $\sigma_i \cdots \sigma_j$, and if it is $\epsilon$ then it produces a ref-word that consists only of variable operations.
\begin{example}
	Figure~\ref{fig:g1} presents a possible parse-tree of a grammar $G_{\doc}$.
\end{example}

We establish the following connection between $G$ and $G_{\doc}$. 
\def \lemgdnonter
{	For every extraction grammar $G$ in CNF, every document $\doc \df \sigma_1\cdots \sigma_n$, every non-terminal $A$ of $G$, and every ref-word $\br\in(\Sigma\cup \Gamma_X)^*$ with $\clean(\br) = \sigma_i \cdots \sigma_j$ the following holds:
	$A\Rightarrow_G^*\br$ 
	if and only if 		 $A_{i,j} \Rightarrow_{G_{\doc}}^* \br$
}
\begin{lemma}\label{lem:gdnonter}
	\lemgdnonter
\end{lemma}  
\begin{proof}
	The proof is by induction on $j-i$. In the base case we have $j=i$, and $\br = \sigma_i$. By definition, $A\Rightarrow^*_G \sigma_i$ if and only if $A_{i,i}\Rightarrow^*_{G_{\doc}} \sigma_i$.
	For the induction step: 
	If $A\Rightarrow^*_{G} \br$ with $\br = \sigma_i \cdots \sigma_j$ 
	then since $G$ is in CNF we have $A\Rightarrow^*_G BC \Rightarrow^*_{G} \br_1 \br_2$ where $\clean(\br_1) = \sigma_i \cdots \sigma_{i+\ell}$ and 
	$\clean(\br_2) = \sigma_{i+\ell+1} \cdots \sigma_{j}$ for some $0\le \ell \le j-i $. 
	By induction hypothesis, $B_{i,i+\ell}\Rightarrow^*_{G_{\doc}} \br_1$ and  $C_{i+\ell+1,j}\Rightarrow^*_{G_{\doc}} \br_2$. By definition of $G_{\doc}$ it holds that $A_{i,j}\Rightarrow^*_{G_{\doc}}B_{i,i+\ell}C_{i+\ell+1,j}$, and thus we conclude 
	$A_{i,j}\Rightarrow^*_{G_{\doc}} \br$.
	
	For the other direction, let us assume that $A_{i,j} \Rightarrow^*_{G_{\doc}} \br$. By $G_{\doc}$'s definition and since it is in CNF, it holds that $A_{i,j} \Rightarrow^*_{G_{\doc}} B_{i, i+\ell}C_{i+\ell+1,j}$ where $A \Rightarrow^*_{G} BC$. Thus, we obtain  $A_{i,j} \Rightarrow^*_{G_{\doc}} B_{i, i+\ell}C_{i+\ell+1,j}\Rightarrow^*_{G_{\doc}}\br_1 \br_2 = \br$. By induction hypothesis, $ B\Rightarrow^*_{G}\br_1$ and $ C\Rightarrow^*_{G}\br_2$. Combining the above we conclude $A\Rightarrow^*_{G} BC \Rightarrow^*_{G} \br_1 \br_2 = \br$.
\end{proof}
We can conclude the following.
\def \lemdecwd {
	For every extraction grammar $G$ in CNF and for every document $\doc$, it holds that $\refl(G,\doc) = \mathcal{L}(G_{\doc})$.
}
\begin{corollary}\label{lem:decwd} 
	\lemdecwd
\end{corollary}

We note that adjusting our extraction grammar to $\doc$ is somewhat  similar to the CYK algorithm~\cite{DBLP:books/daglib/0016921} and therefore it is valid on extraction grammars $G$ in CNF. 
For  a similar reason, we obtain the following complexity which is cubic in $|\doc|$.  
\def\propcompo{
	For every extraction grammar $G$ in CNF and for every document $\doc$,	it holds that $G_{\doc}$ can be constructed in $O({|\doc|}^3|G|)$.} 
\begin{proposition} \label{prop:compo}
	\propcompo
\end{proposition}
\begin{proof}
	For each production rule in $G$, there are at most $|\doc|^3$ production rules in $G_{\doc}$. Thus, the total time required to produce $G_{\doc}$ is $O(|\doc|^3|G|)$.		
\end{proof}

Similar ideas used by Earley's algorithm~\cite{earley1970efficient} might be used to decrease the complexity of constructing $G_{\doc}$. 

\begin{figure}
\begin{center}
\begin{minipage}[b]{0.45\textwidth}
		\centering
		\begin{tikzpicture}[
			scale=0.9, 
			level/.style={sibling distance=1.8cm,
				level distance = 1.2cm}
			]
			\tikzset{level 1/.style={sibling distance=3.4cm}}
			\tikzset{level 2/.style={sibling distance=1.9cm}}
			\tikzset{level 3/.style={sibling distance=1.2cm}}
			\tikzset{frontier/.style={align}}
			\node {${ A_{1,2}}$}
			child{
				node (A) {{${B_{1,1}}$}}
				child {
					node {${D_{ 1,1}}$}
					child{ node{${H_{\epsilon}}$}
						child{ node{$\vop{x}$}}
					}
					child{ node{${I_{1,1}}$}
						child{ node{$\sigma_1$}
						} 
					}
				}
				child {
					node {${E_\epsilon}$}
					child{ node{${J_\epsilon}$}
						child{ node{$\vcl{x}$}}
					}
					child{ node{${K_\epsilon}$}
						child{ node{$\vop{y}$}
						} 
					}
				}
			}
			child{ 
				node {${C_{2,2}}$}
				child{ node{${F_\epsilon}$}
					child{ node{$\vcl{y}$}}
				}
				child{ node{${ G_{2,2}}$}
					child{ node{$\sigma_2$}
					} 
				}
			}
			;
		\end{tikzpicture}
		\caption{
			\label{fig:g1} 
   A parse-tree of the grammar obtained after adjustment to $\doc \df \sigma_1 \sigma_2$
}
\end{minipage}
\hfill
\begin{minipage}[b]{0.45\textwidth}
		\centering
		\begin{tikzpicture}[
			scale=0.9, 
			level/.style={sibling distance=1.8cm,
				level distance = 1.2cm},
                    box/.style = {draw,blue,inner sep=2pt,rounded corners=2pt}] 
			]
			\tikzset{level 1/.style={sibling distance=3.4cm}}
			\tikzset{level 2/.style={sibling distance=1.9cm}}
			\tikzset{level 3/.style={sibling distance=1.2cm}}
			\tikzset{frontier/.style={align}}
			\node (A) {${ A_{1,2}^{\vop{x},\emptyset}}$}
			child{
				node (B) {{${B^{\vop{x},\vcl{x}\vop{y}}_{1,1}}$}}
				child {
					node {${D^{\vop{x},\emptyset}_{ 1,1}}$}
					child{ node{${H_{\epsilon}}$}
						child{ node{$\vop{x}$}}
					}
					child{ node{${I^{\emptyset,\emptyset}_{1,1}}$}
					child{node{$\sigma_1$}}
					}
				}
				child {
					node {${E_\epsilon}$}
     					child{ node{${J_\epsilon}$}
						child{ node{$\vcl{x}$}}
					}
					child{ node{${K_\epsilon}$}
						child{ node{$\vop{y}$}
						} 
					}
				}
			}
			child{ 
				node (C) {${C_{2,2}^{\vcl{y},\emptyset}}$}
				child{ node{${F_\epsilon}$}
					child{ node{$\vcl{y}$}}
				}
				child{ node{${ G_{2,2}^{\emptyset,\emptyset}}$}
				child{node{$\sigma_2$}}
				}
			}
			;
                   \node[box,fit=(A)(B)(C)] {};
		\end{tikzpicture}
		\caption{
			\label{fig:g} A parse-tree obtained after Step 1 of constructing $\decgrmr{G_{\doc}}$. Within the frame, after Step 2.
		}
\end{minipage}
\end{center}
\end{figure}

\subsection{Constructing the Decorated Grammar}

The goal of this step of preprocessing is to encode the information on the produced variable operations within the  non-terminals. 
We obtain from $G_{\doc}$, constructed in the previous step, a new grammar, namely $\decgrmr{G_{\doc}}$, that produces \emph{decorated words} which are sequences of elements of the form
$A_{i,j}^{\bx,\by}$ where $A_{i,j}$ is a non-terminal in $G_{\doc}$ and $\bx,\by$ are sets of variable operations. 
Intuitively, $A_{i,j}^{\bx,\by}$ indicates that all variable operations in $\bx$ occur right before $\sigma_i$, and all those in $\by$ right after $\sigma_j$.

To define $\decgrmr{G_{\doc}}$, we need $G$ to be functional. The following key observation is used in the formal definition of $\decgrmr{G_{\doc}}$ and is based on the functionality of $G$.
\def\propfuncset
{ 	For every functional extraction grammar $G$ and every non-terminal $A$ of $G$, there is a set $\bx_A\subseteq \Gamma_X$ of variable operations such that for every ref-word $\br$ where $A \Rightarrow^* \br$ the variable operations that appear in $\br$ are exactly those in $\bx_A$. Computing all sets $\bx_A$ can be done in $O(|G|)$. }
\begin{proposition}\label{prop:funcset}
	\propfuncset
\end{proposition}
\begin{proof}
	Assume by contradiction that there are $\br, \br'$ that contain different variable operations such that $A\Rightarrow^* \br$ and $A\Rightarrow^* \br'$. 
	Since $G$ does not contain any unreachable non-terminal, and since every non-terminal of $G$ is reachable from the initial symbol there exists a leftmost derivation (that is obtained by applying production to the leftmost non-terminal in each step. ) of the form:
	\[ S\Rightarrow^* \br_1 A \gamma \Rightarrow^* \br_1 \br \br_2 
	\]
		where $\br_1,\br_2 \in (\Sigma\cup \Gamma_X)^*$ and $\gamma \in (\nter \cup \Sigma \cup \Gamma_X )^*$.
	Thus, there is also the following derivation
	\[ S\Rightarrow^* \br_1 A \gamma \Rightarrow^* \br_1 \br' \br_2. 
	\]
	Recall that due to functionality both $\br_1 \br \br_2$ and $\br_1 \br' \br_2$ are valid, and thus  the variable operations that occur in $\br$ and $\br'$ must be the same, which leads to the desired contradiction.
	
	To compute the sets $\bx_A$, we view the grammar as a directed graph whose nodes are the non-terminals and edges are the pairs $(A,B), (A,C)$ whenever there is a production $A\rightarrow BC$, and  $(A,\sigma)$ whenever there is a production $A\rightarrow \sigma$ with $\sigma \in \Gamma_X$.
	We run a topological order on the graph, and iterate over the output of this order in an inverse order. 
	We then run a BFS on this graph and view the output sequence of non-terminals as a topological order. We then iterate over the non-terminals in an inverse order and accumulate for each the set of variable operations that consists of its descendants. 
	This whole process can be done in $O(|G|)$. 
\end{proof}
In other words, for functional extraction grammars, the information on the variable operations is stored implicitly in the non-terminals. The set $\bx_A$ will serve us in the construction of $\decgrmr{G_{\doc}}$.

\subsubsection{Steps of the construction}
The grammar $\decgrmr{G_{\doc}}$ is obtained from $G_{\doc}$ in two steps:
\paragraph*{Step 1} For all pairwise disjoint subsets $\bx,\by,\bz,\bw \subseteq \Gamma_X$, we add the following derivations to $\decgrmr{G_{\doc}}$
\begin{enumerate}
	\item	$A^{\emptyset,\emptyset}_{i,i}\rightarrow \sigma_i$ for every rule $A_{i,i}\rightarrow \sigma_{i}$ in $G_{\doc}$,
	\item	$A_{\epsilon} \rightarrow \tau$	for every rule $A_{\epsilon} \rightarrow \tau$ in $G_{\doc}$ with $\tau \in \Gamma_X$,
	\item $A_{\epsilon} \rightarrow  B_{\epsilon} C_{\epsilon}$ for every  rule	$A_{\epsilon} \rightarrow  B_{\epsilon} C_{\epsilon}$  in $G_{\doc}$,
	\item 
	$A^{\bx,\by\cup \bx_C}_{i,j} \rightarrow  B^{\bx,\by}_{i,j} C_{\epsilon} $ for every rule $A_{i,j} \rightarrow  B_{i,j} C_{\epsilon}$
	in $G_{\doc}$ with $\bx\cap \bx_C = \by \cap \bx_C = \emptyset$ and $\bx_C$ defined as in Proposition~\ref{prop:funcset},
	\item 
	$A^{\bx\cup \bx_B,\by}_{i,j} \rightarrow B_{\epsilon}   C^{\bx,\by}_{i,j}$ for every rule	$A_{i,j} \rightarrow  B_{\epsilon} C_{i,j}$  in  $G_{\doc}$  with  $\bx\cap \bx_B = \by \cap \bx_B = \emptyset$ and $\bx_B$ defined as in Proposition~\ref{prop:funcset},
	\item 
	$A^{\bx,\bw}_{i,j} \rightarrow  B^{\bx,\by}_{i,i'} C^{\bz,\bw}_{i'+1,j}$
	for every rule $A_{i,j} \rightarrow  B_{i,i'} C_{i'+1,j}$ in $G_{\doc}$, and
 \item $S \rightarrow S_{1,n}^{\bx,\by}$ where $S$ is a fresh initial symbol.
\end{enumerate} 
We complete this step by eliminating useless non-terminals (i.e., those that do not produce a terminal string or are not reachable from the initial symbol). We denote the resulting grammar by $G'_{\doc}$.
\begin{example}\sloppy{
	Figure~\ref{fig:g} illustrate the parse-tree obtained after applying Step 1 in the construction of the decorated grammar $\decgrmr{G_{\doc}}$ of $G_{\doc}$ from Figure~\ref{fig:g1}.
	For convenience, in this and in following examples, we present the superscript as pairs of comma-separated sequences, each consists
 of elements in the corresponding set.
Figure~\ref{fig:g} is obtained by propagating information on the variable operations in superscripts. In particular, which occur before, and which after the corresponding subtree.  }
\end{example}
\paragraph*{Step 2}
This step utilizes the following key notion. 
\begin{definition} \sloppy{
A non-terminal $A^{\bx,\by}_{i,j}$ 
is said to be \emph{stable} if $\bx_A =  \bx \cup \by $ (where $\bx_A$ is as defined in Proposition~\ref{prop:funcset}). }
\end{definition}
Intuitively, if a non-terminal $A_{i,j}^{\bx,\by}$ is stable then every ref-word produced by its corresponding non-terminal $A_{i,j}$ in $G_{\doc}$ consists of the sequence of elements in $\bx$, followed by the sequence $\sigma_1 \cdots \sigma_j$, followed by the sequence of elements in $\by$. That is, the only variable operations produced by $A_{i,j}^{\bx,\by}$ are those in $\bx \cup \by$. This intuitively implies that we do not need to complete the derivation of stable non-terminals because their superscripts and subscripts contain all the relevant information on the corresponding mapping.

Following this intuition, we change further the grammar obtained in the previous step by omitting all derivations whose left-hand sides are stable non-terminals. 
We then again eliminate useless non-terminals.

\begin{example}
	The stable non-terminals of the grammar depicted in Figure~\ref{fig:g} are $B_{1,1}^{\vop{x}, \vcl{x},\vop{y}}$ and $C_{2,2}^{\vcl{y}, \emptyset}$, $D_{1,1}^{\vop{x},\emptyset}, G_{2,2}^{\emptyset,\emptyset}$ and $I_{1,1}^{\emptyset,\emptyset}$.
 Indeed, $\bx_B = \{\vop{x}, \vcl{x}, \vop{y} \}$, $\bx_C=\{ \vcl{y} \}$, $\bx_D = \{\vop{x} \}$, and $ \bx_G = \bx_I= \emptyset$.
 The non-terminal $A_{1,2}^{\vop{x},\emptyset}$ is non-stable since $\bx_A = \Gamma_{\{x,y\}} \ne \{\vop{x}\}$.
 The result of applying Step 2 is the subtree that contains $A_{1,2}^{\vop{x},\emptyset}, B_{1,1}^{\vop{x}, \vcl{x},\vop{y}}$ and $C_{2,2}^{\vcl{y}, \emptyset}$ that appears within the frame.
\end{example}

\begin{figure}
		\centering
		\begin{tikzpicture}[
			level/.style={sibling distance=1.4cm,
				level distance = 1.3cm}
			]
			\tikzset{frontier/.style={distance from root=8cm}}
			\node (S){${S^{\emptyset,\emptyset}_{1,n}}$}
			child{
				node (A) {{${A_{1,1}^{\emptyset,\emptyset}}$}}
			}
			child{ 
				node {${B_{2,n}^{\emptyset,\emptyset }}$}
				child{ node{${C_{2,2}^{\emptyset,\emptyset}}$}
				}
				child{ node{${\cdots}$}	
					child{node{${E_{n-2,n-2}^{\emptyset,\emptyset}}$}}
					child{node{$ {F^{\emptyset,\emptyset}_{n-1,n}}$}
						child{ node{${G^{\emptyset,\vop{x}\vcl{x}}_{n-1,n-1}}$}}
						child{		node{${H_{n,n}^{\emptyset,\emptyset}}$ }
						}
					}
				}
			}
			;
		\end{tikzpicture}
		\caption{
			A partial parse-tree of $\decgrmr{G_{\doc}}$ whose depth is linear in the length of the document and whose right branch, except $H_{n,n}^{\emptyset,\emptyset}$, consists of non-stable non-terminals.
		}
		\label{fig:depth}
\end{figure}

The complexity is discussed in the following lemma.
\newcommand{\lemstb}{For every functional extraction grammar $G$ in CNF and for every document $\doc$, the set of stable non-terminals of ${G'_{\doc}}$ is computable in $O(|G_{\doc}|5^{2k})$  where $k$ is the number of variables $G$ is associated with.}
\begin{lemma}\label{lem:stb}
	\lemstb
\end{lemma}
\begin{proof}
	Recall that in Proposition~\ref{prop:funcset} we showed that we can compute the sets $\bx_A$ for every non-terminal $A$ in  $O(|G|)$.
	We can then scan the non-terminals of ${G'_{\doc}}$, and for each such non-terminal $A_{i,j}^{\bx , \by}$ check in $O(1)$ whether $\bx_A = \bx \cup \by$.
	Therefore, the complexity of this step is 
	$O(|G|)+O(|{G'_{\doc}}|) = 
	O(|\doc|^3|G|\,5^{2k})$.
\end{proof}

\subsection{Characterizing the Decorated Grammar}

The decorated grammar $\decgrmr{G_{\doc}}$ encodes information on the original grammar $G$, and is used in the output stage of the algorithm. 
We point at its characteristics which allow us to establish its connection to $G$.

\sloppy{
Notice that the terminals of $\decgrmr{G_{\doc}}$ are stable non-terminals of $G'_{\doc}$. We now define the mapping that is defined by the words produced by $\decgrmr{G_{\doc}}$.
We define the {$(X,d)$-mapping} $\mu^{w}$ that corresponds with a decorated word }
\[
w\df {(A^1)}^{\bx_1,\by_1}_{i_1,j_1}\cdots{(A^m)}^{\bx_m,\by_m}_{i_m,j_m}
\] 
as follows:
$\mu^{w}(x) \df \mspan{i}{j}$ whenever the two following statements hold: 
\begin{itemize}
    \item there is $1\le \ell\le m$ such that either (a) $\vop{x} \in \bx_{\ell}$ and $i=i_{\ell}$ or (b)
$\vop{x} \in \by_{\ell}$ and $i=j_{\ell}+1$;  
\item 
there is $1\le \ell \le m$ such that either (a) $\vcl{x}\in\bx_{\ell}$ and $j= i_{\ell}$ or (b) $\vcl{x}\in\by_{\ell}$ and $j=j_{\ell}+1$.
\end{itemize}
We say that a decorated word $w$ is \emph{valid} if $\mu^{w}(x)$ is well-defined for every $x\in X$.

The connection between $G_{\doc}$ and $\decgrmr{G_{\doc}}$ is as follows.
\def \lemdecwd {
	For every functional unambiguous extraction grammar $G$ in CNF and for every document $\doc$, it holds that every decorated word produced by $\decgrmr{G_{\doc}}$ is valid and that 
	\[\repspnr{G}(\doc) = \{\mu^w \vl S \Rightarrow^*_{\decgrmr{G_{\doc}}} w \}.\]
}
\begin{lemma}\label{cor:decwd} 
	\lemdecwd
\end{lemma}
\begin{proof}

By Corollary~\ref{lem:decwd}, it suffices to show that
\[
\{ \mu^{\br} \vl S_{1,n} \Rightarrow^*_{G_{\doc}} \br \} =  \{\mu^w \vl S \Rightarrow^*_{\decgrmr{G_{\doc}}} w\}
.\]
\newcommand{\tr}{\mathcal{T}}
\newcommand{\dectr}{\tr_{\sc{dec}}}
\noindent Recall that we denote by $G'_{\doc}$ the grammar obtained after Step 1. 

\paragraph{Claim 1}
Let $\br$ be a word over $\Sigma \cup \Gamma_X$ with $\clean{(r)} = \sigma_i\cdots \sigma_j$ and let us denote it by 
$x_1 \cdots x_k \ \sigma_i \ \tau_1\cdots \tau_{k'} \ \sigma_j\  y_1  \cdots y_{k''} $
where $k,k',k''\ge 0$, 
$x_1,\ldots,x_k, y_1,\ldots,y_{k''} \in \Gamma_X$ and 
$\tau_1,\ldots,\tau_{k'} \in \Gamma_X \cup \Sigma$.
It holds that $A_{i,j}^{\bx,\by}\Rightarrow^*_{G'_{\doc}} \br $ 
with $\bx = \{x_1,\ldots,x_k \}$, $\by = \{y_1,\ldots, y_{k''}\}$
if and only if $A_{i,j}\Rightarrow^*_{G_{\doc}} \br$ where both derivations are terminal derivations.

The claim is proved by induction on the derivation length using the definition of $G'_{\doc}$ and the fact it does not contain useless non-terminals. In particular, if the derivation is of length $0$ then the derivation is not terminal and the claim is trivial. Assume that $A_{i,j}^{\bx,\by}\Rightarrow^{n+1}_{G'_{\doc}} \br $ with $n\ge 0$. Then by the definition of $G'_{\doc}$ we can distinguish between the following cases:
\begin{itemize}
    \item $A_{i,j}^{\bx,\by}\Rightarrow^{1}_{G'_{\doc}} \br $ which implies that $i=j$, $\bx=\by=\emptyset$ and $\br = \sigma_i$. Thus, by the definition of $G'_{\doc}$, it holds that  $A_{i,j}^{\bx,\by}\Rightarrow^{1}_{G_{\doc}} \br $ as well.
    \item
    $A_{i,j}^{\bx,\by}\Rightarrow^{1}_{G'_{\doc}} B_{i,j}^{\bx,\bz} C_{\epsilon} \Rightarrow^n_{G'_{\doc}} \br$. We denote $B_{i,j}^{\bx,\bz} \Rightarrow^n_{G'_{\doc}} \br_1$ and $C_{\epsilon} \Rightarrow^n_{G'_{\doc}} \br_2$  (and $\br = \br_1\br_2$). 
    By induction hypothesis, $\br_1$ is of the form 
    $x_1\cdots x_k \sigma_i w \sigma_j z_1\cdots z_{k'}$ where $\bx = \{x_1,\ldots,x_k\}$, $\bz = \{z_1,\ldots,z_{k'}\}$,  $\sigma_i w \sigma_j\in(\Sigma\cup{\Gamma_X})^*$ and $\clean(\sigma_i w \sigma_j) = \sigma_i\cdots\sigma_j$.
    Notice that $\br_2 = \tau_1 \cdots \tau_{k''}$ consists only of variable operations,  in particular, those in $\bx_{C}$, that is $\bx_C = \{\tau_1,\ldots,\tau_{k''}\} $. 
    Thus, 
    by item 4 of the definition of $G'_{\doc}$ we can conclude that $\by = \bz \cup \bx_{C}$ which completes this case. \item
    $A_{i,j}^{\bx,\by}\Rightarrow^{1}_{G'_{\doc}}  B_{\epsilon} C_{i,j}^{\bz,\by} \Rightarrow^n_{G'_{\doc}} \br$ is show symmetrically to the previous case. 
    \item
      $A_{i,j}^{\bx,\bw}\Rightarrow^{1}_{G'_{\doc}} B_{i,i'}^{\bx,\by} C_{i'+1,j}^{\bz,\bw} \Rightarrow^n_{G'_{\doc}} \br$. We denote $B_{i,i'}^{\bx,\by} \Rightarrow^n_{G'_{\doc}} \br_1$ and
      $C_{i'+1,j}^{\bz,\bw} \Rightarrow^n_{G'_{\doc}} \br_2$ (and $\br = \br_1\br_2$).
      By induction hypothesis, 
      $\br_1$ is of the form $x_1\cdots x_k \sigma_i \gamma_1 \sigma_{i'} y_1\cdots y_{k'}$ and 
      $\br_2$ is of the form $z_1\cdots z_m \sigma_{i'+1} \gamma_2 \sigma_{j} w_1\cdots w_{m'}$ where $\gamma_1,\gamma_2\in (\Sigma\cup\Gamma_X)^*$ and $\clean(\gamma_1) = \sigma_i\cdots\sigma_{i'}$, $\clean(\gamma_2)= \sigma_{i'+1}\cdots\sigma_j$. Thus, $\br$ is of the form $x_1\cdots x_k \sigma_i \gamma_1 \sigma_{i'} y_1\cdots y_{k'} z_1\cdots z_m \sigma_{i'+1} \gamma_2 \sigma_{j} w_1\cdots w_{m'}$, which completes the proof.
\end{itemize}

To state the next claim we define a function $M$ that maps sequences of stable non-terminals into sets of pairs of the form $(\tau, i)$ where $\tau$ is a variable operation and $i$ an integer. We define $M$ inductively by $M(\epsilon) \df \emptyset$ and $M(A_{i,j}^{\bx,\by} \gamma) \df  \{ (x,i),(y,j+1) \mid x\in \bx, y\in\by \} \cup M(\gamma)$ where $\gamma$ is a (possibly empty) sequence of non-stable terminals.
The intuition is that $M$ encodes parts of the mapping induced by the sub-ref-word (i.e., the sequence of stable non-terminals) it operates on . 

In addition, we define a function $N$ that maps pairs $(i,\br)$ where $1\le i \le n+1$ is an integer and $\br \in (\Sigma\cup\Gamma_X)^*$ inductively as follows $N(i,\epsilon) \df \emptyset$, $N(i,\sigma \gamma) \df N(i+1,\gamma)$ when $\sigma \in \Sigma$, and 
$N(i,\tau \gamma) \df \{(i,\tau) \} \cup N(i,\gamma)$ when $\tau \in \Gamma_X$.
The intuition is that given an offset $i$ and a partial ref-word $\gamma$, $N(i,\gamma)$ encodes parts of the mapping induced by $\gamma$ with an offset $i$. For instance, if $\gamma = \sigma \tau \sigma'$ where $\sigma,\sigma'\in\Sigma$ and $\tau\in\Gamma_X$ then $N(i,\gamma)$ contains a single pair $(i+1,\tau)$ indicating that in the mapping it encodes the variable operation $\tau$ appears in position $i+1$.
\paragraph{Claim 2}
It holds that 
$A_{i,j}^{\bx,\by}\Rightarrow^*_{\decgrmr{G_{\doc}}} w $ if and only if
$A_{i,j}^{\bx,\by}\Rightarrow^*_{G'_{\doc}} \br $ 
and $M(w) = N(i,\br)$ where both derivations are terminal derivations.
We prove the claim by induction on the length of the derivation of $\decgrmr{G_{\doc}}$.
If the derivation is of length zero then since it is a terminal derivation, it holds that $A_{i,j}^{\bx,\by}$ is stable. 
Hence, $w = A_{i,j}^{\bx,\by}$ and $M(w)=\{ (x,i),(y,j+1) \mid  x\in \bx , y\in \by \} $. Assume $A_{i,j}^{\bx,\by}\Rightarrow^*_{G'_{\doc}} \br $. Since $A_{i,j}^{\bx,\by}$ is stable, Claim 1 implies that $A_{i,j}\Rightarrow^*_{G_{\doc}} \br$ with $\br$ of the form 
$
x_1 \cdots x_k \ \sigma_i \cdots \sigma_j\  y_1  \cdots y_{k''} 
$
where $k,k''\ge 0$, and
$\bx = \{x_1,\ldots,x_k \}$, $\by = \{y_1,\ldots, y_{k''}\}$. Hence $M(w) = N(i,\br)$.
If the derivation is of length greater or equal to one then we can denote $A_{i,j}^{\bx,\by}\Rightarrow^1_{\decgrmr{G_{\doc}}}B_{i,i'}^{\bx,\bw}C_{i'+1,j}^{\bz,\by}\Rightarrow^*_{\decgrmr{G_{\doc}}} w_1 w_2$. By the induction hypothesis, $B_{i,i'}^{\bx,\bw}\Rightarrow^*_{G'_{\doc}} \br_1$ and $C_{i'+1,j}^{\bz,\by}\Rightarrow^*_{G'_{\doc}} \br_2$ and $M(w_1)=N(i,\br_1), M(w_2)=N(i'+1,\br_2)$. By the  properties of $M$ and $N$ we have that $M(w_1w_2)=M(w_1)\cup M(w_2) = N(i,\br_1) \cup N(i'+1,\br_2) = N(i,\br_1\br_2)$ which completes the proof.

Combining the claims we get that $S_{1,n}\Rightarrow^*_{G_{\doc}} \br$ if and only if there are $\bx,\by$ such that $S_{1,n}^{\bx,\by}\Rightarrow^*_{\decgrmr{G_{\doc}}} w$ where $M(w)=N(1,\br)$. By definition, $M(w) = \mu^w$ and $N(1,\br)=\mu^{\br}$ which completes the proof.
\end{proof}
\noindent
Another desired property of the decorated grammar that will help us in the enumeration step is the following. 
\def \propunamb {For every functional extraction grammar $G$ in CNF and for every document $\doc$, if $G$ is unambiguous then 
	$\decgrmr{{G}_{\doc}}$ is unambiguous.}
\begin{proposition}\label{prop:unamb}
	\propunamb
\end{proposition}
\begin{proof}
	It is straightforward that if $G$ is unambiguous then so is $G_{\doc}$ for any document $\doc$.
Combining  Claim 1 from the proof of Lemma~\ref{cor:decwd} and 
the fact $\decgrmr{G_{\doc}}$ is obtained from $G'_{\doc}$ by removing derivation rules, allow us to conclude the desired claim. 
\end{proof}

Finally, combining Proposition~\ref{prop:unamb}  and Lemma~\ref{cor:decwd} leads to the following direct conclusion.
\begin{corollary}
	For every functional unambiguous extraction grammar $G$ in CNF and for every document $\doc$,
	enumerating mappings in $\repspnr{G}(\doc)$ can be done by enumerating the decorated words in $\{w \vl S \Rightarrow^*_{\decgrmr{G_{\doc}}} w \}$.
\end{corollary}

To summarize the complexity of constructing $\decgrmr{G_{\doc}}$ we have:
\def\propcomp{For every functional unambiguous extraction grammar $G$ in CNF and for every document $\doc$, $\decgrmr{G_{\doc}}$ can be constructed in $O(|G_{\doc}|\,5^{2k})=O({|\doc|}^3|G|\,5^{2k})$ where $k$ is the number of variables associated with $G$.
} 
\begin{proposition} \label{prop:comp}
	\propcomp
\end{proposition}
\begin{proof}
	To analyze the complexity of Step 1, we remark that there are $5^{2k}$ four pairwise disjoint subsets of $\Gamma_X$ 
	and thus, the size of the resulting grammar is  $O(|G_{\doc}|5^{2k})$ and this is also the time complexity required for constructing it. 
	Step 2 includes identifying the stable non-terminals which can be done in $O(|G_{\doc}|5^{2k})$ due to Lemma~\ref{lem:stb}, and consists of a constant number of linear scans of $G'_{\doc}$, and hence can be done in $O(|G_{\doc}|5^{2k})$ time. 
	Thus, the total complexity of the construction of $\decgrmr{G_{\doc}}$ is $O(|G_{\doc}|\,5^{2k}) = O({|\doc|}^3|G|\,5^{2k})$.
\end{proof}

Before moving to the output stage of the algorithm, we discuss the ideas that allow us to obtain constant delay between every two consecutive outputs.

\subsection{The Jump Function}
When constructing the decorated grammar, we explained why we can stop the derivation when reaching stable non-terminals; It turns out, that our algorithm can also skip parts of the derivation.

Notice that if $G$ is associated with $k$ variables, there are exactly $2k$ variable operations in each ref-word produced by $G$. While this allows us to obtain an upper bound for the number of non-stable non-terminals in a parse-tree, their depth can be linear in  $|\doc|$. We demonstrate this in the following example. 
\begin{example}
	Consider the non-stable non-terminal $F_{n-1,n}^{\emptyset,\emptyset}$ in the partial parse-tree in Figure~\ref{fig:depth}. Observe that the depth of this  non-terminal is linear in $n \df |\doc|$.
\end{example}

To obtain a delay independent of $\doc$, we may
skip parts of the parse-tree in which no variable operation occurs. 
This idea somewhat resembles the one used by Amarilli et al.~\cite{DBLP:conf/icdt/AmarilliBMN19} in their constant delay enumeration algorithm for regular spanners represented as vset-automata. There, they defined a function that `jumps' from one state to the other if the path from the former to the latter does not contain any variable operation.
We extend this idea to extraction grammars by defining the notion of skippable productions.

Intuitively, for each non-terminal in a parse-tree the output mapping is affected either by its left subtree, or by its right subtree, or by the production applied on this non-terminal itself (or by any combination of the above). 
If the mapping is affected exclusively by the left (right, respectively) subtree then we can skip the production applied on this non-terminal and move to check the left (right, respectively) subtree; We can continue recursively until we reach a production for which this is no longer the case (that is, the mapping is affected by more than one of the above). 

Formally, a \emph{skippable} production rule is of the form 
$A_{i,j}^{\bx, \by} \rightarrow B^{\bx,\emptyset}_{i,i'} C^{\emptyset,\by}_{i'+1,j}$ 
where
\begin{itemize}
	\item[$(a)$] $A_{i,j}^{\bx, \by}$ is non-stable,
and
	\item[$(b)$] exactly one of $B^{\bx,\emptyset}_{i,i'},C^{\emptyset,\by}_{i'+1,j}$ is stable.
\end{itemize}
Intuitively, $(a)$ assures that the parse-tree rooted in $A_{i,j}^{\bx,\by}$ has an effect on the mapping; The empty sets in the superscripts of $B$ and $C$ assure that the production applied on $A_{i,j}^{\bx, \by}$ does not have an effect on the mapping and $(b)$ assures that exactly one subtree of $A_{i,j}^{\bx, \by}$ (either the one rooted at $B^{\bx,\emptyset}_{i,i'}$ if $C^{\emptyset,\by}_{i'+1,j}$ is stable, or the one rooted at $C^{\emptyset,\by}_{i'+1,j}$ if $B^{\bx,\emptyset}_{i,i'}$ is stable) has an effect on the mapping. 
We then say that a skippable production rule $\rho$ \emph{follows} a skippable production rule $\rho'$ if the non-stable non-terminal in the right-hand side  of $\rho'$ (i.e., after the $\rightarrow$) is the non-terminal in the left-hand side  of $\rho$ (i.e., before the $\rightarrow$).
We denote the set of non-terminals of $\decgrmr{G_{\doc}}$ by $\nterdec$, and define  
the  function $\jmp \colon \nterdec \rightarrow 2^{\nterdec}$ as follows:
$B\in \jmp(A_{i,j}^{\bx,\by})$ if there is a sequence of skippable production rules 
$\rho_1,\ldots, \rho_m$ such that:
\begin{itemize}
	\item  $\rho_{\iota}$ follows $\rho_{\iota-1}$ for every $\iota$,
	\item the left-hand side of $\rho_1$ is $A_{i,j}^{\bx,\by}$,
	\item the non-stable non-terminal in the right-hand side of $\rho_m$ is $B$, 
	\item there is a production rule that is non-skippable whose left-hand side is $B$.
\end{itemize}
\begin{example}
Figure~\ref{fig:depth} illustrates one of the parse-trees of a decorated grammar for which  $F_{n-1,n}^{\emptyset,\emptyset}\in \jmp(S_{1,n}^{\emptyset,\emptyset})$. In this parse-tree,  all left children are stable non-terminals. The root $S$ and all the right children except $H$ are non-stable. Intuitively, this allows us to ``jump'' from $S$ directly to $F$, which is the first point we reach a derivation that affects the mapping. This intuition is captured by the definition of the jump function.
\end{example}

The acyclic nature of the decorated grammar (that is, the fact that a non-terminal cannot derive itself) enables us to obtain the following upper bound for the computation of the $\jmp$ function.
\newcommand{\lemjmp}{\sloppy{For every functional unambiguous extraction grammar $G$ in CNF and for every document $\doc$, 
	the $\jmp$ function is computable in $O(|\doc|^5 3^{4k} |G|^2)$  where $k$ is the number of variables $G$ is associated with.}}
\begin{lemma}\label{lem:jmp}
	\lemjmp
\end{lemma}
\begin{proof}
	We claim that Algorithm~\ref{alg:jmp} computes $\jmp$: it gets the decorated grammar $\decgrmr{G_{\doc}}$ and the set $\sf{skippable}(G_{\doc})$ of $\decgrmr{G_{\doc}}$'s skippable productions as input,
	and it outputs the
	array $\sf{jmp}$ for which  $\sf{jmp}[A,i,j,\bx,\by]$ is
	$\jmp(A_{i,j}^{\bx,\by})$. 
	\begin{algorithm}[t]

		\SetKw{Output}{output}
		\SetKw{Initialize}{initialize}
		\SetKw{Denote}{denote}
		\SetKw{Procedure}{procedure}
		\Procedure{$\tsc{computeJump}\big(\decgrmr{G_{\doc}},\sf{skippable}(G_{\doc})\big) $}
		\ForEach{non-terminal $A_{i,j}^{\bx,\by}$ in $\decgrmr{G_{\doc}}$}
		{ $\sf{reachable}[A, i,j,\bx,\by] :=\{ A_{i,j}^{\bx,\by} \}$} 
		
		\ForEach{non-terminal $A$ in $\decgrmr{G_{\doc}}$}
		{ insert $A$ to $\sf{remove}$}  
		
		\ForEach{production rule $\rho$ in $\decgrmr{G_{\doc}}$}
		{
			\If{$\rho$ is non-skipabble with left-hand side $A_{i,j}^{\bx,\by}$}
			{
				$\sf{remove} := \sf{remove} \setminus\{ A_{i,j}^{\bx,\by}\}$  
		}				}
		
		\ForEach{non-terminal $A_{i,j}^{\bx,\by}$ in $\decgrmr{G_{\doc}}$}
		{$\sf{jmp}[A, i,j,\bx,\by] :=\{A_{i,j}^{\bx,\by}\} \setminus \sf{remove}
			$} 
		\Initialize{ \text{perform a topological sort on non-terminals}\;}
		{	\ForEach{ skippable $\rho \df A_{i,j}^{\bx,\by} \rightarrow B_{i,\ell}^{\bx,\emptyset} C_{\ell+1,j}^{\emptyset,\by} $ in $\decgrmr{G_{\doc}}$ in reverse topological order on their left-hand sides}
			{
				\If{ $B_{i,\ell}^{\bx,\emptyset}$ is non-stable}
				{$\sf{reachable}{[A,i,j,\bx,\by]}:= \sf{reachable}[A,i,j,\bx,\by] \cup \sf{reachable}[B,i,\ell,\bx,\emptyset]$\;
					$\sf{jmp}{[A,i,j,\bx,\by]}:= \sf{reachable}[A,i,j,\bx,\by] \setminus \sf{remove}$\;}
				\If{$C_{\ell+1,j}^{\emptyset,\by}$ is non-stable}
				{$\sf{reachable}[A,i,j,\bx,\by] \df \sf{reachable}[A,i,j,\bx,\by] \cup \sf{reachable}[C,\ell+1,j,\emptyset, \by]$
					$\sf{jmp}{[A,i,j,\bx,\by]}:= \sf{reachable}[A,i,j,\bx,\by] \setminus \sf{remove}$\;}
			}
			
		}
		\Output{$\sf{jmp}$}
		\caption{Compute The Jump \label{alg:jmp}}
 \end{algorithm}
	\sloppy{The algorithm uses an auxiliary array $\sf{reachable}$, and an auxiliary set $\sf{remove}$. }
	The correctness of the algorithm is based on the two following claims:
	
	\paragraph*{Claim 1}
	The array $\sf{reachable}$ stores in cell $\sf{reachable}[A,i,j,\bx,\by]$ all of the non-terminals $B$ such that there is a sequence of skippable productions  
	$\rho_1,\ldots, \rho_m$ 
	such that:
	\begin{itemize}
		\item  $\rho_{\iota}$ follows $\rho_{\iota-1}$ for every $\iota$,
		\item the left-hand side of $\rho_1$ is $A_{i,j}^{\bx,\by}$,
		\item the non-stable non-terminal in the right-hand side of $\rho_m$ is $B$ 
	\end{itemize}

	\paragraph*{Claim 2}
	The set $\sf{remove}$ contains 
	all of the non-terminals $B$ 
	for which there is no non-skippable production with $B$ in the left-hand side. (That is, all of the productions that have $B$ in their left-hand side are skippable.)

	These two claims allow us to conclude that if $B \in \sf{reachable}[A,i,j, \bx,\by]$ and $B\not \in \sf{remove}$, 
	there is a sequence of skippable productions  
	$\rho_1,\ldots, \rho_m$ such that:
	\begin{itemize}
		\item  $\rho_{\iota}$ follows $\rho_{\iota-1}$ for every $\iota$,
		\item the left-hand side of $\rho_1$ is $A_{i,j}^{\bx,\by}$,
		\item the non-stable non-terminal in the right-hand side of $\rho_m$ is $B$, 
		\item there is \emph{no} production rule that is non-skippable whose left-hand side is $B$.
	\end{itemize}
	This, in turn, completes the algorithm's proof of correctness.
	
	We shall now explain why both claims hold.
	Claim 2 is derived directly from the way we initialize $\sf{remove}$, and omit elements from it in the third \emph{for-each} loop.
	
 The proof of Claim 1 is more involved.
 Notice that $\sf{reachable}$ is initialized with all the non-terminals of the grammar. 
Let us focus on the last \emph{for-each} loop in which we update $\sf{reachable}$.
Notice that we iterate over the skippable derivations in reverse topological order on their left-hand sides.
We denote by 
$m_{k}$ the last iteration for which the left-hand side of $\rho$ is the element in the $k$th position.  
We show that at the end of iteration $m_{k}$ it holds that if $A_{i,j}^{\bx,\by}$ is the element in the $k$th position then $\sf{reachable[A,i,j,x,y]}$ is updated correctly.  
For the induction base $k = 1$ and the claim holds due to the initialization. 
For the step, assume the claim holds for $m_{k}$. Let us denote the left-hand side of skippable rule of iterations 
$m_{k}+1,\ldots, m_{k+1}$ 
by $A_{i,j}^{\bx,\by}$ which is the element in the $k$th position.
In this case, all skippable rules are of the form 
$\rho \df A_{i,j}^{\bx,\by}\rightarrow B_{i,\ell}^{\bx,\emptyset} C_{\ell+1,j}^{\emptyset, \by}$ with exactly one of $B_{i,\ell}^{\bx,\emptyset} C_{\ell+1,j}^{\emptyset, \by}$ being non-stable.
Notice that in these iterations we go over all these skippable rules with $A_{i,j}^{\bx,\by}$ on their left-hand sides. Indeed, we add  to  $\sf{reachable[A,i,j,x,y]}$
all elements in $\sf{reachable[B,i,\ell,\bx,\emptyset]}$ in case $B_{i,\ell}^{\bx,\emptyset}$ is non-stable, and $\sf{reachable[C,\ell+1,j,\emptyset,\by]}$ if $C_{\ell+1,j}^{\emptyset, \by}$ is non-stable. Since we do so for all skippable rules with $A_{i,j}^{\bx,\by}$ on their left-hand side, and since by induction hypothesis $\sf{reachable[B,i,\ell,\bx,\emptyset]}$ and  $\sf{reachable[C,\ell+1,j,\emptyset,\by]}$ are updated correctly (since they must be in position $<k$ due to $\rho$), the conditions in the claim holds and we can conclude the proof.

	\paragraph*{Complexity}
	We can check in $O(1)$ whether a non-terminal is stable. To do so, we can run the algorithm described in the proof of Lemma~\ref{lem:stb} and store the data in a lookup table with non-terminals as keys and value which is either `stable' or `non-stable'.
	The runtime of this procedure is $O(|G_{\doc}|5^{2k})$.
	In addition, since we can check in $O(1)$ whether a non-terminal is stable or not, we can check in $O(1)$ whether a rule in $\decgrmr{G_{\doc}}$ is skippable just by verifying that the conditions in the definition holds (i.e., its left-hand side is non-stable and exactly one of the non-terminals in its right-hand side is stable). To do so, we can also use the above look-up table, and for each rule check whether all of the conditions of the definition hold in $O(1)$. 
	We can store this information in a lookup table with the rules as keys and the value is either `skippable' or `non-skippable'.
	The runtime of this procedure is also $O(|G_{\doc}|5^{2k})$.
	
	The  first four \emph{for-each} loops are used for initialization and  require $O(|G_{\doc}|5^{2k})= O(|\doc|^3 |G| 5^{2k})$ altogether as they iterate through all production rules of $\decgrmr{G_{\doc}}$.
	%
	The topological sort of the production rules requires $O(|\doc|^3 |G| 5^{2k})$.

	In the last \emph{for-each} loop, we iterate through all skippable rules of the decorated grammar $\decgrmr{G_{\doc}}$ ordered by their head in reverse topological order.
	There are at most $O(|\doc|^3|G| 3^{2k})$ such rules (since their form is restricted to $A_{i,j}^{\bx,\by} \rightarrow B_{i,\ell}^{\bx,\emptyset} C_{\ell+1,j}^{\emptyset, \by} $).
	In each such iteration, we 
	compute set unions and set difference. 
	The size of each of these sets is bounded by $O(|{\doc}|^2 3^{2k}|G|)$.
	Thus, the total time complexity  of the second for loop is $O(|{\doc}|^5 3^{4k}|G|^2)$.
	And finally, the total complexity of the algorithm is 
	$O(|{\doc}|^5 3^{4k}|G|^2)$.
\end{proof}

According to Lemmas~\ref{lem:stb} and~\ref{lem:jmp} the overall computation time required to find the 
non-stable non-terminals as well as to compute the $\jmp$ function is quintic in the document size and can therefore be included in the preprocessing stage. 

It is important to note that if we can reduce the complexity of computing the $\jmp$  function to cubic then we can reduce the whole preprocessing time to cubic.

\section{Output Stage of the Enumeration Algorithm}\label{sec:en}

In the output stage of our algorithm, we build recursively the parse-trees of the decorated grammar $\decgrmr{G_{\doc}}$ that was constructed in the preprocessing stage.

\begin{figure}[t]
		\centering
		\begin{tikzpicture}[ 
			level/.style={sibling distance=1.1cm,
				level distance = 1.0cm},
     box/.style = {draw,blue,inner sep=4pt,rounded corners=2pt}
			]
			\tikzset{level 1/.style={sibling distance=3.4cm}}
			\tikzset{level 2/.style={sibling distance=2.1cm}}
   \tikzset{level 3/.style={sibling distance=1.9cm}}
    \tikzset{level 4/.style={sibling distance=1.5cm}}
			\tikzset{frontier/.style={distance from root=8cm}}
			\node (S){${S^{\vop{x}\vop{y},\vcl{z}}_{1,6}}$}
			child{
				node (A) {{${A_{1,4}^{{\vop{x} \vop{y}, \vop{z}}}}$}}
				child {
					node (C) {${C_{1,3}^{{\vop{x} \vop{y}, \vcl{y} }  }}$} 
				}
				child {
					node (D){${D_{4,4}^{\emptyset,{\vop{z}}}}$}
				}
			}
			child{ 
				node (B) {${B_{5,6}^{\emptyset, \vcl{z}}}$}
				child{ node (E){${E_{5,5}^{\emptyset, \emptyset }}$}
				}
				child{ node (F) {${F_{6,6}^
						{\vcl{x},\vcl{z}}
						}$}
				}
			}
			;
		\end{tikzpicture}
		\caption{
			 $\decgrmr{G_{\doc}}$  parse-tree.   }
		\label{fig:tree}
	\end{figure}%
 
\subsection{The Output Stage Algorithm}

The enumeration procedure is presented in Algorithm~\ref{alg:enum}.  
The recursive $\enum$ procedure
represents the
$(X,\doc)$-mappings $\mu$ as sets of pairs $(\vop{x},i), (\vcl{x},j)$ whenever $\mu(x) = \mspan{i}{j}$. This representation is useful as $\enum$ gradually constructs the output mappings during the execution. 

\newcommand{\myalg}
{	\begin{algorithm}[H]
		\SetKw{Output}{output}		
		\SetKw{Denote}{denote}
  		\SetKw{Set}{set}
		\SetKw{Procedure}{procedure}
		\SetKw{Assign}{assign}
		\Procedure
		{$\enum(
			\alpha , \map) $}
		
		\If{ $\alpha = \epsilon$} 
		{\Output{$\map$}\;}
		\If{$\alpha = S_{1,n}^{\bx,\by}$ \text{ and } $S_{1,n}^{\bx,\by}$ \text{ is stable} }
		{\Output{$\map$}\; }
		\Denote{$\alpha$ \textbf{by} $A \cdot \alpha'$}\;
		\ForEach{$B\in\jmp{[A]}$}
		{	
			\Assign{ $B' \df \applyProd( B )$}\;
			\ForEach{  $(\beta, \map') \in 
				B'$ }
			{
				$\enum(\beta \cdot \alpha', \map \cup \map')$;
			}
		}
		\caption{\label{alg:enum} Output Stage}
	\end{algorithm}
}
\begin{figure}[h]
	\centering
	\begin{minipage}[b]{0.8\textwidth}
		\myalg
	\end{minipage}
\end{figure}

Before discussing $\enum$, we discuss the procedure $\applyProd$ it calls. 
Intuitively, $\applyProd$ reflects the effect of applying production rules $\rho$ on the non-terminal it gets as input.
The returned pairs $(\beta,\map)$ are such that $\beta$ consists of those non-terminals in the right-hand side of $\rho$ that have an effect on the output mapping, and $\map$ consists of the pairs added to the output mapping as a result of applying $\rho$.

\begin{algorithm}[h]
	\SetKw{Return}{return}
	\SetKw{Procedure}{procedure}
	\SetKw{Output}{output}
	\Procedure{$\applyProd(A_{i,j}^{\bx,\by}) $}
	\\
	initialize $\map=\emptyset$\;  
	
	\ForEach{non-skippable production of the form $A_{i,j}^{\bx,\by}\rightarrow B^{\bx,\bz}_{i,\ell}C^{\bw,\by}_{\ell+1,j}$  }
	{
		
		$\map = \emptyset$\;
		
		\ForEach{ $x\in \bz \cup \bw $}
		{$\map = \map \cup \{ (x, {\ell+1}) \}$}

		\If{ $B^{\bx,\bz}_{i,\ell}$ and $C^{\bw,\by}_{\ell+1,j}$ are non-stable}
		{$\beta = B^{\bx,\bz}_{i,\ell} C^{\bw,\by}_{\ell+1,j}  $}
		\If{ $B^{\bx,\bz}_{i,\ell}$ is non-stable and $C^{\bw,\by}_{\ell+1,j}$ is stable}
		{$\beta = B^{\bx,\bz}_{i,\ell}  $}
		\If{ $B^{\bx,\bz}_{i,\ell}$ is  stable and $C^{\bw,\by}_{\ell+1,j}$ is non-stable}
		{$\beta = C^{\bw,\by}_{\ell+1,j} $}
		\If{ $B^{\bx,\bz}_{i,\ell}$  and $C^{\bw,\by}_{\ell+1,j}$ are stable}
		{$\beta = \epsilon $}
		
		\Output$(\beta , \map)$
	}	
	\caption{Apply Production\label{alg:applyprod} }
\end{algorithm}

\begin{example}
	The procedure $\applyProd$ applied on $S_{1,6}^{\vop{x}\vop{y}, \vcl{z}}$ from Figure~\ref{fig:tree} adds the pair $(\vop{z},5)$ to $\map$; When applied
	on $A_{1,4}^{\vop{x} \vop{y},\vop{z}}$, adds the pair $(\vcl{y},4)$ to $\map$; When applied
	on $B_{5,6}^{\emptyset,\vcl{z}}$, adds the pair $(\vcl{x},6)$ to $\map$.
\end{example}

$\enum$ recursively constructs output mappings by calling the procedure $\applyProd$ to apply productions on non-stable non-terminals while skipping the skippable productions using $\jmp$ that was computed in the preprocessing stage for every non-terminal $A$ and is stored in $\jmp[A]$.
Its input $(\alpha,\map)$ is such that $\alpha$ is a sequence of non-stable non-terminals (i.e., those non-terminals that affect the output mapping) that need to be derived, and $\map$ which accumulates the pairs in the output mapping. 


The main enumeration algorithm calls sequentially
the recursive procedure $\enum$ with pairs $(S_{1,n}^{\bx,\by}, \map)$ for all non-terminal of the form $S_{1,n}^{\bx,\by}$ in $\decgrmr{G_{\doc}}$ with $\map$ that consists of  pairs $(\tau,1)$ for all $\tau\in \bx$, and $(\tau,n+1)$ for all $\tau\in \by$.


\newcommand{\thmenumenum}{For every functional unambiguous extraction grammar $G$ in CNF and for every document $\doc$, the main enumerating algorithm described above enumerates the mappings in $\repspnr G (\doc)$ (without repetitions) with delay of $O(k)$ between each two consecutive mappings where $k$ is the number of variables $G$ is associated with.}
\begin{theorem}\label{thm:enumenum}
	\thmenumenum
\end{theorem}
Had $G$ been ambiguous, the complexity guarantees on the delay would not have held since the same output might have been outputted arbitrary many times.

Finally, we conclude this section by remarking  that the proof of Theorem~\ref{thm:enum} follows from 
Corollary~\ref{cor:decwd}, Proposition~\ref{prop:comp}, Lemma~\ref{lem:stb}, Lemma~\ref{lem:jmp}, and Theorem~\ref{thm:enumenum}.
In the rest of this section we discuss the proof of  Theorem~\ref{thm:enumenum}.

\subsection{Proof of Theorem~\ref{thm:enumenum}}

In order to proof the theorem, we need to take a careful look at different stages of the execution of $\enum$ which correspond with different stages of derivations. 
To do that,  
	we define the mapping $\mu^{\alpha}$ that corresponds with a finite sequence  $\alpha$ of terminals of $\decgrmr{G_{\doc}}$: $\mu^{\alpha}$ is defined by $(x,i)\in \mu^{\alpha}$ whenever $\alpha$ contains the terminal $A_{i,j}^{\bx,\by}$ with $x\in \bx$ or the terminal $A_{i',i-1}^{\bx,\by}$
 with $x\in\by$.

\newcommand{\si}{S_{1,n}^{\bx,\by}}
\newcommand{\ai}{A_{i,j}^{\bx,\by}}
\newcommand{\bci}{B_{i,i'}^{\bx,\bz} C_{i'+1, j}^{\bw,\by}}
\newcommand{\abc}{A_{i,j}^{\bx,\by} \rightarrow B_{i,i'}^{\bx,\bz} C_{i'+1, j}^{\bw,\by}}
\newcommand{\stb}{\tsc{nonStable}}

We first wish to show how the run of $\enum$ reflects the derivation of the decorated grammar. 
\begin{lemma}
The call $\enum(\ai\alpha,\map)$ with $\map \supseteq \{ (x,i) \vl x\in \bx  \} \cup \{ (y,j+1) \vl y\in \by \}$ provokes a call to
$\enum(\alpha,\map \cup \mu^w)$ whenever there is a terminal derivation $\ai \Rightarrow^* w$. 
\end{lemma}
\begin{proof}
	We first observe that 
for every call 
$\enum(A_{i,j}^{\bx,\by} \alpha,\map)$ 
if $A_{i,j}^{\bx,\by}$ is stable then 
$i=1,j=n+1$ and $A=S$. 
Indeed, assume there is a call $\enum(A_{i,j}^{\bx,\by} \alpha,\map)$ 
with a stable $A_{i,j}^{\bx,\by}$. If $A_{i,j}^{\bx,\by} = S_{1,n}^{\bx,\by}$ then we are done. If not, this call was provoked by another call to $\enum$. This implies that there is a pair $(\beta, \map') \in B'$ that is not stable. Nevertheless, $B'$ is returned by $\applyProd$ and due to its definition this leads to a contradiction.

We prove the claim by induction on the length of the terminal derivation. 
The basis holds trivially. For the induction step,  
since $\ai$ is not stable, we can
distinguish between the following two possible cases:

If $\ai$ is skippable then 
$\jmp(\ai)$ returns $B_{i',j'}^{\bx',\by'}$ such that there is a leftmost derivation $\ai \Rightarrow^* \alpha' B_{i',j'}^{\bx',\by'}$ where $\mu^{\alpha'} = \emptyset$ (this can be formally proved by induction on the number of productions applied in the derivation). Thus, $\enum(B_{i',j'}^{\bx',\by'} \alpha, \map)$ is provoked and the claim follows from the induction hypothesis.
 
If $\ai$ is non-skippable then $\jmp(\ai)$ returns $\ai$ itself, and for each production of the form $\abc$ it holds that $\applyProd(\ai)$ returns  the pair $(\beta, \{ (x,i'+1) \,\vl \,i'\in \bz\cup \bw  \})$ where $\beta$ is either $B_{i,i'}^{\bx,\bz}$ or $C_{i'+1,j}^{\bw,\by}$ or their concatenation. In case $\beta$ is $B_{i,i'}^{\bx,\bz}$ then the next call is to 
$\enum(B_{i,i'}^{\bx,\bz} \alpha, \map \cup \{ (x,i'+1) \,\vl \,i'\in \bz\cup \bw  \})$. We can then apply the induction hypothesis and obtain that there is a call to $\enum(\alpha, \map \cup \{ (x,i'+1) \,\vl \,i'\in \bz\cup \bw  \} \cup \mu^{\beta})$ where $B_{i,i'}^{\bx,\bz} \Rightarrow^* \beta$. In this case, we have $\ai \Rightarrow^* B_{i,i'}^{\bx,\bz} C_{i'+1,j}^{\bw,\by} \Rightarrow^* \beta \gamma $ such that $\gamma$ is of the form 
$D_{i'+1,j}^{\bw,\emptyset}$.
Thus, $\mu^{\gamma} = \emptyset$ which implies that $\mu^{\beta \gamma }= \mu^{\beta}$ which, in turn, completes the proof of this case. 
In case $\beta =  C_{i'+1,j}^{\bw,\by}$ we show the claim similarly. And in case $\beta = B_{i,i'}^{\bx,\bz} C_{i'+1,j}^{\bw,\by}$ we use twice the induction hypothesis in a similar way. 
\end{proof}
This leads to the following direct conclusion which can be seen as proof of soundness of $\enum$.
\begin{corollary}\label{cor:dir1}
	The call $\enum(S_{1,n}^{\bx,\by},\map)$ with $\map \supseteq \{ (x,1) \vl x\in \bx  \} \cup \{ (y,n+1) \vl y\in \by \}$ returns $\mu^{w}$ whenever there is a terminal derivation $S_{1,n}^{\bx,\by} \Rightarrow^* w$. 
\end{corollary}

To show that $\enum$ is also complete, that is, for each terminal derivation of the decorated grammar there is a run with corresponding output, we prove the following.
\begin{lemma}
	For every terminal derivation $A \Rightarrow^* \alpha$, the call to the procedure $\enum(A \theta,\map)$ 
	invokes once  $\enum(\theta,\map \cup \mu^{\alpha})$.
\end{lemma}
\begin{proof}
We prove this claim by induction on the number of non-skippable production rules that were applied in a leftmost derivation $A \Rightarrow^* \alpha$.

For the basis, if all of the production applied were skippable this implies that $A$ was stable which implies that $A=S_{1,n}^{\bx,\by}$ and that $\theta = \epsilon$ and thus the claim holds (see the second if clause in $\enum$).
 
Assume that there is at least one production rule that was applied and is non-skippable. 
Thus, there is the following leftmost derivation
\[
A  \Rightarrow^* \alpha \ai \delta\Rightarrow \alpha B_{i,i'}^{\bx,\bz} C_{i'+1,j}^{\bw,\by} \delta \Rightarrow^*\alpha \beta \gamma \delta
\] 
where  the productions applied in $A \Rightarrow^* \alpha \ai \delta$ are all skippable, and the one applied in $\alpha \ai \delta\Rightarrow \alpha B_{i,i'}^{\bx,\bz} C_{i'+1,j} \delta $ (i.e., $\abc$) is non-skippable. 
Since all of the productions that were applied in $A \Rightarrow^* \alpha \ai \delta$ are skippable, we can show by a simple induction (on the number of productions applied) that $\mu^{\alpha} = \emptyset$. 
By the definition of $\enum$, $\applyProd(\ai)$ is invoked and then $\enum(B_{i,i'}^{\bx,\bz} C_{i'+1,j} \delta,\map )$ is invoked (recall that $\mu^{\alpha} = \emptyset$).
By induction hypothesis (on $B_{i,i'}^{\bx,\bz} \Rightarrow^* \beta)$, we conclude that $\enum( C_{i'+1,j} \delta, \map \cup \mu^{\beta} )$ is invoked. Applying again the induction hypothesis (this time on $C_{i'+1,j}^{\bw,\by} \Rightarrow^* \gamma)$ shows that  $\enum( \delta, \map \cup \mu^{\beta} \cup \mu^{\gamma})$ is invoked.

Notice that the call to $\enum(\theta, \map \cup \mu^{\alpha})$ is single due to the unambiguous nature of the decorated grammar.
\end{proof}

This leads us to the following direct conclusion that shows the desired completeness.
\begin{corollary}\label{cor:dir2}
	For every terminal derivation $S_{1,n}^{\bx,\by} \Rightarrow^* \alpha$ the call to the procedure $\enum(S_{1,n}^{\bx,\by} ,\map)$ 
    with $\map \df \{ (x,1),(y,n+1) \, \vl \, x\in \bx, y\in\by \}$ 
	outputs $\mu^{\alpha}$ once.
\end{corollary}

Combining Corollaries~\ref{cor:dir1} and~\ref{cor:dir2} leads to the following key lemma which intuitively show that $\enum$ outputs exactly what is needed.
\begin{lemma}\label{lem:enumout}
The procedure	$\enum(S_{1,n}^{\bx,\by},\map)$ with $\map \df \{ (x,1) \vl x\in \bx  \} \cup \{ (y,n+1) \vl y\in \by \}$ outputs every and each mapping in
	$
	\{ \mu^w \,\,\vline  \,\,
	S_{1,n}^{\bx,\by} \Rightarrow^* w  \}
	$ exactly once.
\end{lemma}
We next discuss the delay between consecutive  outputs of $\enum$, and show that it is independent of the document length. 
\begin{lemma}\label{lem:com}
	Every two consecutive outputs of  $\enum(S_{1,n}^{\bx,\by},\map)$ are different, and 
	the delay between them is bounded by $O(k)$.
\end{lemma}
\begin{proof}
Lemma~\ref{lem:enumout} allows us to conclude the first part of the statement. 
For the delay guarantees, notice that $(1)$ each output of $\enum$ consists exactly of $2k$ pairs. (This is also a straightforward conclusion from the same Lemma.) In addition, $(2)$ 
the procedure $\applyProd$ outputs 
 $(\beta, \map)$ with $\map \ne \emptyset$. This is due to the fact it
  is called on a non-terminal $A$ which was returned by $\jmp$. That is, there is a non-skippable rule with $A$ on the left of $\rightarrow$. Therefore, with the notation used by the algorithm, $\bz \cup \bw \ne \emptyset$ and $\map$ is updated with at least one pair.

Combining $(1)$ and $(2)$ allows us to conclude that between each two outputs of $\enum$, the procedure $\applyProd$ is called at most $2k$ times.
		We can thus conclude that the number of nested calls to $\enum$ between each two consecutive outputs is at most $2k$. Since the computation of $\applyProd$ is linear and since  $\jmp$ is computed as part of the preprocessing we obtain the desired result.
\end{proof}

Finally, to complete the proof of the theorem we combine Corollaries~\ref{cor:dir1} and~\ref{cor:dir2},  Lemma~\ref{lem:com}, and the different calls to $\enum$ the main enumeration algorithm invokes.

\section{Conclusion}\label{sec:conc}
In this paper 
we propose a new grammar-based language for document spanners, namely extraction grammars. We compare the expressiveness of context-free spanners with previously studied classes of spanners and present a pushdown model for these spanners. 
We present an enumeration algorithm for unambiguous grammars that outputs results with a constant delay after quintic preprocessing in data complexity. 
We conclude by presenting recent advances and suggesting several future research directions.

While this paper was being under review, a follow-up paper~\cite{amarilli2022efficient} to the conference version of this paper~\cite{DBLP:conf/icdt/Peterfreund21} was published. This follow-up contains an extension of the enumeration algorithm that allows obtaining constant delay enumeration after a cubic (rather than quintic) preprocessing time. 

To reach a full understanding of the expressiveness of context-free spanners, one should characterize the string relations that can be expressed with context-free spanners. This can be done by understanding
the expressiveness of context-free grammars enriched with string equality selection. We note that
there are some similarities between recursive Datalog over regex formulas~\cite{DBLP:conf/icdt/PeterfreundCFK19} and extraction grammars. Yet, with the former, we reach the full expressiveness of polynomial time spanners (data complexity) whereas with the latter we cannot express string equality. Understanding the connection between these two formalisms better can be a step in understanding the expressive power of extraction grammars.
\OMIT{
Regarding our enumeration complexity, it might be possible to decrease the preprocessing complexity by using other techniques to compute the jump function. 
Another direction is to find
restricted classes of extraction grammars that are more expressive than regular spanners yet allow linear time preprocessing (similarly to~\cite{DBLP:conf/icdt/AmarilliBMN19}).
}

It can be interesting to examine more carefully whether the techniques used here for enumerating the derivations can be applied also for enumerating queries on trees, or enumerating queries beyond MSO on strings.  
This connects to a recent line of work on efficient enumeration algorithms for monadic-second-order queries on trees~\cite{DBLP:conf/pods/AmarilliBMN19}. Can our techniques be used to obtain efficient evaluation for more expressive queries?



  \bibliographystyle{elsarticle-num} 
  \bibliography{references.bib}


%
%
%
\end{document}